\documentclass[11pt]{article}

\usepackage{fullpage}
\usepackage[T1]{fontenc}
\usepackage{amssymb} 
\usepackage{amsmath}  
\usepackage{amsthm}
\usepackage{times}

\usepackage{xcolor}
\usepackage[normalem]{ulem} 

\usepackage{hyperref}

\title{Semijoins of Annotated Relations}

\begin{document}

\author{Phokion G.\ Kolaitis\\ University of California Santa Cruz \\
  kolaitis@ucsc.edu}

\date{}
\maketitle

\newcommand\phk[1]{\textcolor{blue}{(Phokion) #1}}
\newcommand\albert[1]{\textcolor{green}{(Albert) #1}}


\newcommand{\phrcomment}[1]{\marginpar{\tiny \textbf{M:} #1}}
\newcommand{\aacomment}[1]{\marginpar{\tiny \textbf{A:} #1}}

\newtheorem{lemma}{Lemma}
\newtheorem{claim}{Claim}
\newtheorem{corollary}{Corollary}
\newtheorem{theorem}{Theorem}
\newtheorem{proposition}{Proposition}
\newtheorem{definition}{Definition}
\newtheorem{fact}{Fact}
\newtheorem{example}{Example}


\newenvironment{examplebf}{\refstepcounter{example}\par\bigskip
\noindent\textbf{Example~\theexample.} \rmfamily}{\hfill$\dashv$\medskip}

\newcommand{\supp}{{\mathrm{Supp}}}
\newcommand{\tuples}{{\mathrm{Tup}}}
\newcommand{\domain}{{\mathrm{Dom}}}
\newcommand{\Entropy}{{\mathrm{H}}}
\newcommand{\KL}{{\mathrm{D}}}
\newcommand{\ftp}{transportation property}
\newcommand{\icp}{inner consistency property}
\newcommand{\ltgc}{local-to-global consistency property}

\newcommand{\gcpr}{global consistency problem for relations}
\newcommand{\gcpb}{global consistency problem for bags}
\newcommand{\glcpr}{{\sc GCPR}}
\newcommand{\glcpb}{{\sc GCPB}}

\newcommand{\intcone}{\mathrm{intcone}}
\newcommand{\norm}[1]{\Vert #1 \Vert}
\newcommand{\bnorm}[1]{\norm{#1}_{\mathrm{b}}}
\newcommand{\unorm}[1]{\norm{#1}_{\mathrm{u}}}
\newcommand{\suppnorm}[1]{\norm{#1}_{\mathrm{supp}}}
\newcommand{\munorm}[1]{\norm{#1}_{\mathrm{mu}}}
\newcommand{\mbnorm}[1]{\norm{#1}_{\mathrm{mb}}}

\newcommand{\free}[1]{\mathbb{F}(#1)}

\newcommand{\transpose}{\mathrm{T}}
\newcommand{\Id}{\mathrm{I}}
\newcommand{\trace}{\mathrm{tr}}

\newcommand{\onto}{\stackrel{\scriptscriptstyle{s}}\rightarrow}

\newcommand{\standardjoin}[1]{\Join_{#1,\mathrm{S}}}
\newcommand{\vorobyevjoin}[1]{\Join_{#1,\mathrm{V}}}
\newcommand{\componentwisejoin}[2]{\Join^{#1}_{#2}}

\newcommand{\commentout}[1]{}

\newcommand{\icpadj}{innerly consistent}

\newcommand{\newatop}[2]{\genfrac{}{}{0pt}{2}{#1}{#2}}

\newcommand{\wit}{{\mathrm Wit}}

\newcommand{\cj}{\Join_{\mathrm c}}

\newcommand{\cjoin}{${\mathrm c}$-join}

\newcommand{\umon}{uniformly monotone}

\newcommand{\sleq}{\sqsubseteq}

\newcommand{\sjk}{\ltimes_{\mathbb K}}

\newcommand{\sj}{\mathcal S}

\newcommand{\fpp}{production property}

\thanks{}

\begin{abstract}
The semijoin operation is a fundamental operation of relational algebra that has been extensively used in query processing.  Furthermore, semijoins have been used to formulate desirable properties of acyclic  schemas; in particular,  a schema is acyclic if and only if it has a full reducer, i.e., a sequence of semijoins that converts a given collection of relations to a globally consistent collection of relations. In recent years, the study of acyclicity has been  extended to annotated relations, where the annotations are values from some positive commutative monoid. 
So far, however, it has not been known if the characterization of acyclicity in terms of full reducers extends to annotated relations.  Here, we develop a theory of semijoins of annotated relations. To this effect, we first
introduce the notion of a semijoin function on a monoid and then characterize the positive commutative monoids for which a semijoin function exists. After this, we introduce the notion of a full reducer for a schema on a monoid and  show that the following is true for every positive commutative monoid that has the inner consistency property: a schema is acyclic if and only if it has a full reducer on that monoid.
\commentout{
During the early days of relational database theory it was realized that ``acyclic'' database schemas possess a number of desirable  
properties. In fact, three different notions of ``acyclicity'' were identified and 
investigated during the 1980s, namely,  $\alpha$-acyclicity, $\beta$-acyclicity,  and $\gamma$-acyclicity. Much more recently, the study of $\alpha$-acyclicity was extended to annotated relations, where the annotations are values from some positive commutative monoid. 
The recent results about $\alpha$-acyclic schemas and annotated relations give rise  to
results about $\beta$-acyclic schemas and annotated relations, since a schema is $\beta$-acyclic if and only if every sub-schema of it is $\alpha$-acyclic.
Here, we study $\gamma$-acyclic schemas and annotated relations. Our main finding is that the characterization of $\gamma$-acyclic schemas in terms of monotone sequential join expression extends to annotated relations, provided the annotations come from a positive commutative monoid that has the inner consistency property.
Furthermore, the results reported here shed light on the role of the join of two standard relations. Specifically, our results reveal that
the only relevant property of the join of two standard relations is that it is a witness to the consistency of the two relations, provided that these two relations are consistent. For the more abstract setting of annotated relations, this property of the standard join is captured by the notion of a consistency witness function, a notion which we systematically 
utilize in this work.
}

\end{abstract}

\maketitle

\section{Introduction} \label{sec:intro}

The semijoin operation has been  extensively studied since the early days of relational databases, where the semijoin $R\ltimes T$ of two relations $R$ and $T$ is the relation consisting of the tuples in $R$ that join with at least one tuple in $T$.
Bernstein and Chiu \cite{DBLP:journals/jacm/BernsteinC81} and  then Bernstein and Goodman \cite{DBLP:journals/siamcomp/BernsteinG81}
were the first to investigate the semijoin operation as a tool in query processing   over distributed databases. 
The advantage of using semijoins is that, in certain cases, semijoins make it possible to compute the answers to queries while avoiding  intermediate computations of expensive joins. After this, Yannakakis \cite{DBLP:conf/vldb/Yannakakis81} showed that the join of a collection of relations over an acyclic schema can be efficiently computed using semijoins only. In the aforementioned papers, the notion of a \emph{semijoin program} underlies the main results, where a semijoin program is a finite sequence of assignment statements of the form
$R_i:=  R_i\ltimes R_j$, for some relations $R_i$ and $R_j$. Beeri, Fagin, Maier, and Yannakakis \cite{BeeriFaginMaierYannakakis1983} characterized acyclic  schemas in terms of several different desirable semantic properties. 
Two key such properties are the \emph{\ltgc~}~and the existence of a \emph{full reducer}. Recall that two relations $R_1$ and $R_2$ are consistent if there is a relation $W$ such that the projection of $W$ on the attributes of $R_i$ is equal to $R_i$, for $i=1,2$. By definition, a database schema has the \ltgc~if every collection $R_1,\ldots, R_m$ of pairwise consistent relations  over such schemas is globally consistent (i.e.,  there is a relation $W$ whose projection on the attributes of $R_i$ is equal to $R_i$, for $1 \leq i \leq m$). Also by definition, a database schema has a full reducer if there is a semijoin program $\pi$  such that for every collection $R_1,\ldots, R_m$ of relations  over the schema, the output of $\pi$ on input $R_1,\ldots,R_m$ is a globally consistent collection $R_1^*,\ldots,R_m^*$
 of relations. The main result in Beeri et al.\ \cite{BeeriFaginMaierYannakakis1983}
asserts that the following are equivalent for a schema $H$: (i) $H$ is acyclic; (ii)  $H$ has the \ltgc; (iii) H has a full reducer.

During the past two decades, a body of research on annotated databases has been developed, where in an annotated database   each tuple in a relation is annotated with a value from some algebraic structure. Following the influential work on database provenance by Green, Karvounarakis, and Tannen \cite{DBLP:conf/pods/GreenKT07},  several different aspects of annotated databases have been investigated,  including 
the study of conjunctive query containment for annotated databases \cite{DBLP:journals/mst/Green11,DBLP:journals/tods/KostylevRS14}, the evaluation of Datalog programs on annotated databases \cite{DBLP:journals/jacm/KhamisNPSW24},
and the study of the semantics of first-order logic  on annotated databases \cite{DBLP:journals/corr/abs-1712-01980,DBLP:journals/corr/abs-2412-07986}. In this framework,  the annotations are values in the universe of    some fixed semiring ${\mathbb K}=(K,+, \times, 0,1)$; this is a  simultaneous generalization of  standard relational databases, where the annotations are $1$  and $0$,  and of  bag databases, where the  annotations are non-negative integers denoting  multiplicities of tuples. The term $\mathbb K$-relations refers to relations with annotations from a semiring $\mathbb K$; thus, standard relations are $\mathbb B$-relations, where $\mathbb B$ is the Boolean semiring, while bags are $\mathbb N$-relations, where $\mathbb N$ is the semiring of non-negative integers.

Atserias and Kolaitis \cite{AK25} studied the interplay between local consistency and global consistency for  annotated relations by addressing the following question: do the aforementioned results by Beeri et al. \cite{BeeriFaginMaierYannakakis1983}  characterizing acyclicity in terms of the \ltgc~extend  to annotated relations? Since the definition of consistency of  annotated relations uses only the projection operation on relations and since projection is defined using only addition $+$, they considered $\mathbb K$-relations where  the annotations come from a monoid ${\mathbb K=(K, +, 0)}$.
 They then identified the following condition on monoids:  a monoid $\mathbb K$ has the \emph{inner consistency property} if whenever two $\mathbb K$-relations are inner consistent, then they are also consistent, where two $\mathbb K$-relations $R$ and $T$ are inner consistent if the projection of $R$ on the common attributes of $R$ and $T$ coincides with the projection of $T$ on the common attributes of $R$ and $T$ (note that consistency always implies inner consistency).
  The first main result in \cite{AK25} asserts  that the following is true for every  positive commutative monoid $\mathbb K$ that has the inner consistency property:  a schema $H$ is acyclic if and only if $H$ has
 the \ltgc~for $\mathbb K$-relations. The second main result in \cite{AK25} asserts that a positive commutative monoid $\mathbb K$ has the inner consistency property if and only if every acyclic schema has the \ltgc~for $\mathbb K$. 
 In addition to the Boolean monoid and the bag monoid,  a plethora of other monoids turned out to  have the inner consistency property, including all monoids arising in the context of database provenance.

\smallskip

\noindent{\bf Summary of Results}~
In \cite{AK25}, the following   problem was raised: do  the results by Beeri et al. \cite{BeeriFaginMaierYannakakis1983} characterizing  acyclicity in terms of the existence of a full reducer   extend and how do they extend to annotated relations? As stated in \cite{AK25}, the main difficulty is whether a suitable notion of semijoin can be defined on monoids.  Our goal here is to resolve this open problem.

When one first considers this problem, one 
may attempt to define a semijoin operation on $\mathbb K$-relations as  the projection of the ``join'' of two $\mathbb K$-relations $R$ and $T$ on the attributes of $R$, where ${\mathbb K}=(K,+,0)$ is a monoid. 
There are, however, two issues with this approach. First, to define the ``join'' of two $\mathbb K$-relations, one needs both an addition operation $+$ and a multiplication operation $\times$ on elements of $K$. Second, even if the monoid $\mathbb K=(K,+,0)$ has an expansion to a semiring
$(K,+,\times, 0,1)$,  the join operation arising from $+$ and $\times$  need not yield a ``good'' semijoin operation. A striking case in point is the bag semiring ${\mathbb N}=(N,+,\times, 0,1)$ of non-negative integers: as shown in \cite{DBLP:conf/pods/AtseriasK21}, there are consistent bags $R$ and $T$ such that the projection of the bag-join $R\Join T$ on the attributes of $R$ is different from $R$.

To overcome these issues, we adopt an axiomatic approach, that is, we introduce the notion of a \emph{semijoin function  on a monoid ${\mathbb K}=(K,+,0)$} by specifying four structural properties that such a function $\sj$  must obey (see Definition \ref{defn:semijoin} for the details). The first property stipulates that if $R$ and $T$ are two consistent $\mathbb K$-relations, then $\sj(R,T)=R$. 
The other three properties make use of the canonical preorder $\sqsubseteq$ of a monoid, where for all  $a$ and $b$ in $K$, we have that $a\sqsubseteq b$ if there is 
some $c$ in $K$ such that $a+c=b$. Note that all four properties are satisfied by the semijoin of standard relations; thus, our notion of a semijoin function on a monoid generalizes indeed the semijoin operation in relational algebra.

Once the notion of a semijoin function is in place, we explore the existence of semijoin functions. To this effect, we identify a property of monoids, called the \emph{production property}, which has an operations-research character (see the comments following Definition \ref{dfn:frugal}). We then show that a positive commutative monoid $\mathbb K$ has a semijoin function if and only if $\mathbb K$ has the production property.  This result has several fruitful consequences. On the positive side, we show that if $\mathbb K$ has the inner consistency property, then $\mathbb K$ has the production property, hence $\mathbb K$ also has a semijoin function. On the negative side, we show that no numerical semigroup other than  the bag monoid ${\mathbb N}=(N,+,0)$ 
has the \fpp, hence no numerical semigroup other than $\mathbb N$ has a semijoin function (a numerical semigroup is a submonoid of $\mathbb N$ whose universe consists of all but finitely many non-negative integers \cite{assi2020numerical}). Further, we use the production property to show that the following problem is solvable in polynomial-time: given a positive commutative monoid $\mathbb K$, does $\mathbb K$ have a semijoin function?

Semijoin functions on monoids make it possible to give semantics to semijoin programs over monoids.  In turn, this gives rise to the notion of a \emph{full reducer} for a schema $H$ on a monoid $\mathbb K$, i.e., a pair $(\pi, \sj)$ consisting of a semijoin program $\pi$ involving $H$ and a semijoin function $\sj$ on $\mathbb K$ such that for every collection $R_1,\ldots,R_m$ of $\mathbb K$-relation over $H$, the resulting collection $R_1^*,\ldots,R_m^*$ of $\mathbb K$-relations over $H$ obtained by applying $\pi$ and $\sj$ is globally consistent. Our main result asserts that the following is true for every positive commutative monoid $\mathbb K$ that has the inner consistency property: a schema $H$ is acyclic if and only if $H$ has a full reducer on $\mathbb K$. This generalizes the results in Beeri et al.\ \cite{BeeriFaginMaierYannakakis1983}  and yields a host of new results about annotated relations, including new results about bags, since the bag monoid $\mathbb N$ has inner consistency property. Furthermore, we show that if $H$ is an acyclic schema, then a single semijoin program $\pi$ gives rise to a full reducer $(\pi,\sj)$ for every semijoin function $\sj$ on any positive commutative monoid $\mathbb K$ that has the inner consistency property.  

Finally, by combining the preceding results with results in \cite{AK25}, we obtain a new characterization of the inner consistency property by showing that a positive commutative monoid $\mathbb K$ has the inner consistency property if and only if every acyclic schema has a full reducer on $\mathbb K$.

\commentout{
Annotated databases are databases in which each fact in a relation is annotated with a value from some algebraic structure. Starting with the influential work on database provenance \cite{DBLP:conf/pods/GreenKT07,DBLP:journals/sigmod/KarvounarakisG12}, there has been an extensive investigation of several different aspects of annotated databases, including 
the study of conjunctive query containment for annotated databases \cite{DBLP:journals/mst/Green11,DBLP:journals/tods/KostylevRS14} and  the evaluation of Datalog programs on annotated databases \cite{DBLP:journals/jacm/KhamisNPSW24}. In these investigations, the annotations are values  from  some fixed semiring ${\mathbb K}=(K,+, \times, 0,1)$. Thus, standard relational databases are annotated databases  in which the annotations are $1$ (true) and $0$ (false), while bag databases are annotated databases in which the  annotations are non-negative integers denoting the multiplicities. This framework, which is often referred to as \emph{semiring semantics}, has spanned  first-order logic \cite{DBLP:journals/corr/abs-1712-01980} and least fixed-point logic
\cite{DBLP:conf/csl/DannertGNT21}.

During the early days of relational database theory it was realized that ``acyclic'' database schemas possess a number of desirable  semantic  properties. In fact, three different notions of ``acyclicity'' were identified and extensively investigated during the 1980s, namely,  acyclicity (also known as $\alpha$-acyclicity), $\beta$-acyclicity,  and $\gamma$-acyclicity. On undirected graphs (equivalently, on database schemas consisting of 
binary relation symbols only) these notions coincide with the notion of an acyclic graph, but they form a strict hierarchy on hypergraphs (equivalently, on arbitrary database schemas) with $\beta$-acyclicity  being a stricter notion than acyclicity, and $\gamma$-acyclicity  being a stricter notion than  $\beta$-acyclicity.

The study of acyclic schemas was initiated by Yannakakis, who focused on the evaluation of acyclic joins \cite{DBLP:conf/vldb/Yannakakis81}.  After this, Fagin, Beeri, Maier, and Yannakakis \cite{BeeriFaginMaierYannakakis1983}
showed that acyclic schemas are precisely the ones possessing the \emph{\ltgc}, that is, every collection of pairwise consistent relations $R_1,\ldots, R_m$ over such schemas is globally consistent (i.e.,  there is a relation $T$ whose projection on the attributes of $R_i$ is equal to $R_i$, for $1 \leq i \leq m$). Fagin et al.\ \cite{BeeriFaginMaierYannakakis1983} also characterized acyclicity in terms of the existence of \emph{monotone sequential join expressions}, i.e., expressions of the form 
$(((\cdots (R_1\Join R_2) \Join \cdots )\Join R_{m-1} )\Join R_m)$ with the property that if the relations $R_1,\ldots,R_m$ are pairwise consistent, then every intermediate sequential join expression $((\cdots (R_1\Join R_2) \Join \cdots )\Join R_{i-1})$ produces a relation that is consistent with the relation $R_i$.
Results about acyclicity yield results about $\beta$-acyclicity, since a schema is $\beta$-acyclic if and only if every sub-schema of it is acyclic. Fagin \cite{DBLP:journals/jacm/Fagin83} 
studied  $\gamma$-acyclicity  and showed that a schema is $\gamma$-acyclic if and only if every \emph{connected} sequential join expression is monotone. Intuitively,  this means  that every sequential join expression is monotone, provided no join between relations with disjoint sets of attributes is allowed.

Atserias and Kolaitis \cite{AK25} studied the interplay between local consistency and global consistency for  annotated relations. Since the definition of consistency of  annotated relations uses only the projection operation on relations and since projection is defined using only addition $+$, they considered $\mathbb K$-relations where  the annotations come from a monoid ${\mathbb K=(K, +, 0)}$.
 They identified a condition on monoids,  called the 
 \emph{inner consistency property}, and showed that a positive monoid ${\mathbb K}=(K,+,0)$ has the 
 inner consistency property
 if and only if every acyclic schema $H$ has the \ltgc~for $\mathbb K$-relations (i.e., every pairwise consistent collection of $\mathbb K$-relations over $H$ is globally consistent). It was not clear, however, whether  the results about acyclic schemas and sequential join expressions in \cite{BeeriFaginMaierYannakakis1983} can be extended to annotated relations, since, as shown in \cite{DBLP:conf/pods/AtseriasK21}, the analog of the standard join for bags need not  be a witness to the consistency of two consistent bags.
 In a subsequent paper, Atserias and Kolaitis \cite{DBLP:journals/sigmod/AtseriasK25} 
introduced the notion of a \emph{consistency witness function} on a positive monoid ${\mathbb K}$, which is a function $W$ that, given two $\mathbb K$-relations $R$ and $S$, returns a $\mathbb K$-relation $W(R,S)$ that is a consistency witness for $R$ and $S$, provided that $R$ and $S$ are consistent $\mathbb K$-relations. They also introduced the notion of a \emph{monotone sequential \cjoin~expression}, which is analogous to that of a monotone sequential join expression with some arbitrary consistency witness function in place of the standard join. Using these notions, it was shown in \cite{DBLP:journals/sigmod/AtseriasK25} that the characterization of acyclicity in terms of monotone sequential join expressions in \cite{BeeriFaginMaierYannakakis1983} extends  to  characterizations of acyclicity in terms of monotone sequential \cjoin~expressions on monoids  having  the
inner consistency property;
furthermore, the 
inner consistency property
itself can be characterized in such terms.

Here, we investigate $\gamma$-acyclic schemas and establish that
 the desirable semantic properties of $\gamma$-acyclic schemas extend to annotated relations.
 The two main results  are as follows:
    \begin{enumerate}
\item If $\mathbb K$ is a positive commutative monoid and $H$ is a schema 
such that every connected sequential \cjoin-expression over $H$ is monotone on $\mathbb K$ w.r.t.\ \emph{some} consistency witness function on $\mathbb K$, then $H$ is $\gamma$-acyclic.
\item If $\mathbb K$ is a positive commutative monoid that has the 
inner consistency property
and $H$ is a schema which is $\gamma$-acyclic, then 
every connected sequential \cjoin-expression over $H$ is monotone on $\mathbb K$
w.r.t.\ \emph{every} consistency witness function on $\mathbb K$.
    \end{enumerate}
As a byproduct of these two main results, we obtain a characterization of the 
inner consistency property
in terms of $\gamma$-acyclicity and connected sequential \cjoin~expressions. 
 Furthermore, our work sheds light on the role of the join of two standard relations. Specifically, our results reveal that, in the study of the various notions of acyclicity in \cite{DBLP:journals/jacm/Fagin83},
the only relevant property of the join of two standard relations is that it is a witness to the consistency of the two relations, provided  these two relations are consistent. In the setting of  annotated relations, this property of the standard join is captured by the notion of a consistency witness function.

 The rest of the paper is organized as follows. Section 2 contains the definitions of the basic notions, while Section 3 contains the definition of a consistency witness function and related notions. To make the paper as self-contained as possible,  the earlier results about acyclic schemas are summarized in Section 4. Section 5 discusses $\beta$-acyclic schemas. Section~6 contains the main results about $\gamma$-acyclic schemas and annotated relations.
 }
 
\section{Basic Notions} \label{sec:prelims}

\noindent{\bf Monoids}

A \emph{commutative monoid} is a structure $\mathbb{K}=(K,+,0)$, where $+$ is a binary operation on the universe $K$ of $\mathbb K$  that is
 associative, commutative, and has $0$ as its neutral  element, i.e., $p+ 0 = p = 0 + p$ holds for all $p\in K$.  
 A commutative monoid $\mathbb K = (K,+,0)$ is \emph{positive} 
 if for all elements $p,q\in K$ with
$p+q=0$, we have that  $p=0$ and $q=0$.  From now on, we  assume that all commutative monoids considered have at least two elements in their universe.

The following are examples of positive commutative monoids.
\begin{itemize}
\item The \emph{Boolean monoid}
 $\mathbb{B} = (\{0,1\},\vee,0)$ with disjunction $\vee$ as its operation and~$0$~(false) as its neutral element.
 \item The \emph{bag  monoid}  $\mathbb{N}=(N, +, 0)$,  where $N$ is the set of non-negative integers   and $+$ is the standard addition operation. 
 Note that the structure ${\mathbb Z}=(Z,+,0)$, where $Z$ is the set of all integers, is a commutative monoid, but not a positive one. 
 \item A \emph{numerical semigroup} is a submonoid ${\mathbb K}=(K,+,0)$ of  the bag monoid $\mathbb{N}=(N, +, 0)$, such that $K$ is a cofinite set, i.e., the complement $N\setminus K$ is  finite. A concrete example of a numerical semigroup is ${\mathbb K}_{3,5}=(\langle 3,5\rangle, +,0)$, where $\langle 3,5\rangle$ is  the set of all non-negative integers of the form $3m+5n$ with  $m\geq 0$ and $n\geq 0$,  i.e., $\langle 3,5\rangle=\{0,3,5,6,8, 9, 10, \ldots \}$.
  \item The \emph{power set monoid} ${\mathbb P}(A)= (\mathcal{P}(A), \cup,\emptyset)$, where 
  if $A$ is a set, then $\mathcal{P}(A)$ is its powerset, and $\cup$ is the union operation on sets.
 \item 
  The \emph{fuzzy monoid} ${\mathbb V}=([0,1], \max, 0)$, where $[0,1]$ is the interval of all real numbers between $0$ and $1$, and $\max$ is  the standard   maximum operation.

\item The structure ${\mathbb R}^{\geq 0}=([0,\infty), +, 0)$, where $[0,\infty)$
is the set of all non-negative real numbers and $+$ is the standard addition operation.
\item   The structure ${\mathbb T}= ((-\infty,\infty], \min, \infty)$, where $(-\infty,\infty]$ is the set of all real numbers together with $\infty$,  and $\min$ is the standard minimum  operation.
\end{itemize}
\noindent{\bf {$\mathbb K$-relations and marginals of $\mathbb K$-relations}}

An \emph{attribute}~$A$ is a symbol with an associated
set~$\domain(A)$ as its \emph{domain}. If~$X$ is a finite set of
attributes, then~$\tuples(X)$ is the set
of~\emph{$X$-tuples}, i.e., the set of
functions that take each attribute~$A \in X$ to an element of its
domain~$\domain(A)$. $\tuples(\emptyset)$ is non-empty as it
contains the \emph{empty tuple}, i.e., the  function with empty
domain. If~$Y \subseteq X$  and~$t$ is
an~$X$-tuple, then the \emph{projection of~$t$ on~$Y$}, denoted
by~$t[Y]$, is the unique~$Y$-tuple that agrees with~$t$ on~$Y$. In
particular,~$t[\emptyset]$ is the empty tuple.

Let~${\mathbb K} = (K,+,0)$ be a positive commutative monoid and let~$X$ be a finite set
of attributes.
\begin{itemize}
    \item 
A~\emph{$\mathbb{K}$-relation over~$X$} is a
function~$R : \tuples(X) \rightarrow K$ that assigns a value~$R(t)$ in~$K$
to every~$X$-tuple~$t$ in~$\tuples(X)$. 
We will often write $R(X)$ to indicate that $R$ is a $\mathbb K$-relation over $X$, and we will refer to $X$ as the set of attributes of $R$.
If~$X$ is the empty set of attributes, then a~$\mathbb{K}$-relation 
over~$X$ is simply a function that assigns a  single value from~$K$  to the empty tuple. 
Note that the $\mathbb B$-relations are  the standard relations of relational database theory, while the $\mathbb N$-relations are the \emph{bags} or 
\emph{multisets}, i.e., each tuple has a non-negative integer associated with it that denotes the  \emph{multiplicity} of the tuple.
\item 
The \emph{support} $\supp(R)$ of
a~$\mathbb{K}$-relation~$R(X)$ is the set
of~$X$-tuples~$t$ that are assigned non-zero value, i.e.,
$\supp(R) := \{ t \in \tuples(X) : R(t) \not= 0 \}$.
We will often write~$R'$ to
denote~$\supp(R)$. Note that~$R'$ is a standard relation
over~$X$. A~$\mathbb{K}$-relation is \emph{finitely supported} if its support is a
finite set. In this paper, all~$\mathbb{K}$-relations considered will be  finitely supported, 
and we omit the term; thus, from now on, a $\mathbb{K}$-relation is a finitely supported $\mathbb{K}$-relation. 
When~$R'$ is empty, we say that~$R$ is the empty~$\mathbb{K}$-relation over~$X$. 
\item 
If~$Y\subseteq X$, then the \emph{projection} (or the \emph{marginal $R[Y]$}) \emph{of $R$ on $Y$}  is the~$\mathbb{K}$-relation 
over~$Y$ such that for every~$Y$-tuple~$t$, we have that
$R[Y](t) := \sum_{\newatop{r \in R':}{r[Y] = t}} R(r).$

The value $R[Y](t)$ is  the \emph{marginal of $R$ over $t$}. For notational simplicity, we will often write $R(t)$ for the marginal of $R$ over $t$, instead of $R[Y](t)$. It will be clear from the context (e.g., from the arity of the tuple $t$) if $R(t)$ is indeed the marginal of $R$ over $t$ (in which case $t$ must be a $Y$-tuple) or $R(t)$ is the actual value of $R$ on $t$ as a mapping from $\tuples(X)$ to $K$ (in which case $t$ must be an $X$-tuple).
   Note that if $R$ is a standard  relation (i.e., $R$ is a $\mathbb B$-relation), then the projection $R[Y]$ is the standard projection of $R$ on $Y$.
\end{itemize}
   The proof of the next useful proposition follows easily from the definitions.

\begin{proposition} \label{lem:easyfacts1} 
Let $\mathbb K$ be a 
 positive commutative monoid and let $R(X)$ be a  $\mathbb K$-relation. Then the
  following  hold:
  \begin{enumerate} \itemsep=0pt
  \item For all~$Y \subseteq X$, we have~$R'[Y] = R[Y]'$.
  \item For all $Z \subseteq Y \subseteq X$, we have $R[Y][Z] = R[Z]$.
\end{enumerate}
\end{proposition}

\commentout{
\begin{proof}
  For the first part, the inclusion~$R[Y]' \subseteq R'[Y]$ is obvious. For the converse, assume that~$t \in R'[Y]$, so
  there exists~$r$ such that~$R(r) \not= 0$ and~$r[Y] =
  t$. By~\eqref{eqn:marginal} and the positivity of~$\mathbb K$, we have
  that~$R(t) \not= 0$. Hence~$t \in R[Y]'$.  
  
  For the second part, we have
  \begin{equation}
  R[Y][Z](u) = \sum_{\newatop{v \in R[Y]':}{v[Z]=u}} R[Y](v) =
  \sum_{\newatop{v \in R'[Y]:}{v[Z]=u}} \sum_{\newatop{w \in R':}{w[Y]=v}} R(w) =
  \sum_{\newatop{w \in R':}{w[Z]=u}} R(w) = R[Z](u)
\end{equation}
where the first equality follows from~\eqref{eqn:marginal}, the second
follows from the first part of this lemma to replace~$R[Y]'$ by~$R'[Y]$, and
again~\eqref{eqn:marginal}, the third follows from partitioning the
tuples in~$R'$ by their projection on~$Y$, together
with~$Z \subseteq Y$, and the fourth follows from~\eqref{eqn:marginal}
again.
\end{proof}
}

\commentout{
If~$X$ and~$Y$ are sets of attributes, then we write~$XY$ as
shorthand for the union~$X \cup Y$. 

Accordingly, if~$x$ is
an~$X$-tuple and~$y$ is a~$Y$-tuple such
that~$x[X \cap Y] = y[X \cap Y]$, then we write~$xy$ to denote
the~$XY$-tuple that agrees with~$x$ on~$X$ and on~$y$ on~$Y$.  We
say that~\emph{$x$ joins with~$y$}, and that~\emph{$y$ joins
  with~$x$}, to \emph{produce} the tuple~$xy$.}
  
\noindent{\bf Schemas and hypergraphs}
\begin{itemize}
\item A \emph{schema} is a sequence~$X_1,\ldots,X_m$ of non-empty sets of attributes.  
\item A \emph{hypergraph} is a pair $H=(V,F)$, where $V$ is a finite non-empty set and $F$ is a set of non-empty subsets of $V$. We call $V$ the set of the \emph{nodes} of $H$ and we call $F$ the set of the \emph{hyperedges} of $H$.
\end{itemize}
A schema~$X_1,\ldots,X_m$ can be identified with  the hypergraph $H=(\bigcup_{i=1}^m X_i,\{X_1,\ldots,X_m\})$, i.e.,  the nodes of $H$  are the attributes  and  the hyperedges of $H$ are the members~$X_1,\ldots,X_m$ of the schema.
In what follows, 
we will use the terms \emph{schema} and
\emph{hypergraph} interchangeably.
\begin{itemize}
\item A  \emph{collection of
$\mathbb{K}$-relations} over  a  schema~$X_1,\ldots,X_m$
is a sequence $R_1(X_1),\ldots,R_m(X_m)$  such that each~$R_i(X_i)$ is a~$\mathbb{K}$-relation over~$X_i$.
\end{itemize}

\section{Background on Consistency and Acyclicity} \label{sec:acyclic}

Let $\mathbb K=(K,+,0)$ be a positive commutative monoid.
\begin{itemize}
\item Two $\mathbb K$-relations $R(X)$ and $T(Y)$ are \emph{consistent} if there is a $\mathbb K$-relation $W(X\cup Y)$ with $W[X]=R$ and $W[Y]=T$.
Such a $\mathbb K$-relation $W$ is a \emph{consistency witness} for $R$ and $T$.
\item A collection $R_1(X_1),\ldots,R_m(X_m)$ of
$\mathbb{K}$-relations over a  schema
$X_1,\ldots,X_m$ is \emph{globally consistent} if there is  a $\mathbb K$-relation $W(X_1\cup \ldots \cup  X_m)$ such that
$W[X_i]=R_i$, for $i$ with $1\leq i\ \leq m$. Such a $\mathbb K$-relation $W$ is  a \emph{consistency witness} for $R_1,\ldots,R_m$.
\end{itemize}
Note that if $R_1(X_1),\ldots,R_m(X_m)$
is a globally consistent collection of $\mathbb K$-relations, then these relations
are pairwise consistent. Indeed, if $W$ is a consistency witness for
$R_1(X_1),\ldots,R_m(X_m)$, then for all $i$ and $j$ with $1\leq i,j \leq m$, we have that the $\mathbb K$-relation $W[X_i \cup X_j]$ is a consistency witness for $R_i$
and $R_j$, because 
$$
R_i =W[X_i]=W[X_i \cup X_j][X_i] \quad \mbox{and} \quad  
R_j  =W[X_j]=W[X_i \cup X_j][X_j],$$
where, in each case, the first equality follows from the definition of global consistency and the second equality follows from the second part of Proposition \ref{lem:easyfacts1}. 
The converse  fails even for standard relations (i.e., $\mathbb B$-relations). 
For example, consider the
\emph{triangle} schema $\{A,B\}, \{B,C\}, \{C,A\}$ and the  standard relations $R_1(A,B)=\{ (0,0), (1,1)\}$, $R_2(B,C)=\{ (0,1),(1,0)\}$,
$R_3(C,A)=\{(0,0), (1,1)\}$. It is easy to check that $R_1, R_2, R_3$ are  pairwise consistent; however, they are not globally consistent since if they were, then their join $((R_1\Join R_2)\Join R_3)$ would be a consistency witness for them, but $((R_1\Join R_2)\Join R_3) = \emptyset$. 

Beeri et al.\ \cite{BeeriFaginMaierYannakakis1983} characterized the schemas for which every collection of pairwise consistent standard relations is globally consistent by showing that these are precisely the \emph{acyclic} schemas (also known as $\alpha$-\emph{acyclic} schemas). Intuitively, the notion of an acyclic hypergraph generalizes to hypergraphs the property that
 a graph is acyclic if and only if every connected component of it  with at least two edges has an articulation point. We will not give the precise definition of acyclicity here because we will not use it in the sequel. Instead, we will give the precise definitions of two different notions that turned out to be equivalent to acyclicity.

\newpage 

Let $H$ be a hypergraph with $X_1,\ldots,X_m$ as its hyperedges. 
\begin{itemize}
\item $H$ has the \emph{running intersection property} if there is an ordering $Y_1,\ldots,Y_m$ of the hyperedges of $H$ such that for every $i\leq m$, there is a $j<i$ such that $(Y_1\cup \cdots \cup  Y_{i-1})\cap Y_i \subseteq Y_j$.
\item $H$ has the \emph{\ltgc~for standard relations}  if every collection \\$R_1(X_1),\ldots,R_m(X_m)$   of pairwise consistent standard relations over $H$ is also globally consistent.
\end{itemize}


\begin{theorem}[\cite{BeeriFaginMaierYannakakis1983}] \label{thm:BFMY}
For every hypergraph $H$, the following statements are equivalent:
\begin{enumerate}
\item $H$ is acyclic.
\item 
$H$ has the running intersection property.
\item $H$ has the \ltgc~for standard relations.
\end{enumerate}
\end{theorem}

As an illustration of Theorem \ref{thm:BFMY},
the triangle schema $\{A,B\}, \{B,C\}, \{C,A\}$ is cyclic;
therefore, Theorem \ref{thm:BFMY} predicts that the triangle schema does not have the \ltgc~for standard relations, which we showed to be the case earlier. 
For every $n\geq 2$, the  \emph{$n$-path} schema  $P_n$ with hyperedges
$\{A_1,A_2\}, \{A_2,A_3\}, \ldots, \{A_{n},A_{n+1}\}$ is acyclic because it has the running intersection property via this ordering. 
Finally, consider the schema $\{A,B,C\}, \{C,D,E\}, \{E,F,A\}, \{A,C,E\}$. It has the running intersection property via the ordering
$\{A,B,C\}, \{A,C,E\}, \{C,D,E\}, \{E,F,A\}$, hence it is acyclic (and so, by Theorem \ref{thm:BFMY}, it has the \ltgc~property for standard relations).

In \cite{AK25}, the following question was investigated: does  Theorem \ref{thm:BFMY} extend from standard relations to $\mathbb K$-relations, where $\mathbb K$ is an arbitrary positive commutative monoid?
Clearly, the first two statements in Theorem \ref{thm:BFMY} are ``structural'' properties that depend on the hypergraph $H$ only; in contrast, the third statement is a ``semantic'' property, as it also involves standard relations.  The \ltgc~property for standard relations has the following natural generalization to relations over monoids.

Let $\mathbb K$ be a positive commutative monoid.
\begin{itemize}
    \item  We say that a hypergraph $H$ with $X_1,\ldots,X_m$ as its hyperedges  has the \emph{\ltgc~for $\mathbb K$-relations} if every collection $R_1(X_1),\ldots,R_m(X_m)$ of pairwise consistent $\mathbb K$-relations is also globally consistent. 
\end{itemize}
In \cite{AK25}, it was that shown that the acyclicity of a hypergraph $H$ is a necessary, but not always sufficient, condition for $H$ to have the \ltgc~for $\mathbb K$-relations, where $\mathbb K$ is an arbitrary positive commutative monoid. 
Nonetheless, 
it was also shown in \cite{AK25} that   Theorem \ref{thm:BFMY} generalizes to every positive commutative monoid $\mathbb K$ that satisfies
 a condition called the \emph{inner consistency property}.
\begin{itemize}
    \item 
    Two $\mathbb K$-relations $R(X)$ and $T(Y)$ are \emph{inner consistent} if
$R[X\cap Y]=T[X \cap Y]$. 

\item We say that $\mathbb K$ has the \emph{inner consistency property} if whenever two $\mathbb K$-relations are inner consistent,  they  are also consistent.
\end{itemize}
Note that, using Proposition \ref{lem:easyfacts1}, it is easy to verify that if $R$ and $S$ are consistent $\mathbb K$-relations, then they are also inner consistent. 
Consequently, for  monoids with the inner consistency property, the notions of consistency and inner consistency coincide.
In particular, this holds true for the 
the Boolean monoid ${\mathbb B}=(\{0,1\}, \vee, 0)$. 

We can now state one of the main results in \cite{AK25}.
\begin{theorem} [\cite{AK25}]
\label{thm:BFMY-general}
Let $\mathbb K$ be a positive commutative monoid that has the inner consistency property.  For every hypergraph $H$,  
 the following statements are equivalent:
\begin{enumerate}
\item $H$ is acyclic.
\item $H$ has the local-to-global consistency property for $\mathbb K$-relations.
\end{enumerate}

\end{theorem}

The bag monoid ${\mathbb N}=(N,+,0)$ was shown to have  the inner consistency property
in \cite{DBLP:conf/pods/AtseriasK21}.
 Several other types of  monoids were shown to  have the inner consistency property in  \cite{AK25}; concrete examples include  the monoids
${\mathbb V}=([0,1], \max, 0)$, ${\mathbb R}^{\geq 0}=([0,\infty), +, 0)$,
${\mathbb T}= ((-\infty,\infty], \min, \infty)$, and the power set monoids ${\mathbb P}(A)= (\mathcal{P}(A), \cup,\emptyset)$, for every set $A$.  In contrast,  no numerical semigroup other than the bag monoid ${\mathbb N}=(N,+,0)$ has the inner consistency property; in particular, the monoid  ${\mathbb K}_{3,5}=(\langle 3,5\rangle, +,0)$ lacks this property.

As shown in \cite{AK25}, the inner consistency property actually characterizes the monoids $\mathbb K$ for which every acyclic hypergraph has the \ltgc~for~$\mathbb K$-relations.

\begin{theorem} [\cite{AK25}]
\label{thm:TP}
 Let $\mathbb K$ be a  positive commutative monoid. Then the following statements are equivalent:
\begin{enumerate} 

\item $\mathbb K$ has the inner consistency property.

\item Every acyclic hypergraph has the \ltgc~for $\mathbb K$-relations.

\item The $3$-path hypergraph $P_3$ has the \ltgc~for $\mathbb K$-relations.

\end{enumerate}

\end{theorem}

\section{Semijoin Functions}
By definition, the \emph{semijoin} $R \ltimes T$ of two standard relations $R(X)$
and $T(Y)$ is the standard relation over $X$ consisting of all tuples in $R(X)$ that join with at least one tuple from $T(Y)$. Thus, $R\ltimes T = (R\Join T) [X]$, that is to say, the semijoin of $R$ and $T$ is the projection of the join of $R$ with $T$ on the set of attributes of $R$.

In this section, we address the question: is there a notion of a semijoin of two $\mathbb K$-relations $R(X)$ and $T(Y)$, where $\mathbb K$ is a positive commutative monoid? To this effect, 
 we introduce the notion of a \emph{semijoin function on $\mathbb K$} by identifying certain desirable properties that such a function ought to obey. We  give  an  explicit  semijoin function for a certain class of monoids and then
 characterize the positive commutative monoids for which a semijoin function exists. Furthermore, we
 prove that if a positive commutative monoid $\mathbb K$ has the inner consistency property, then a semijoin function on $\mathbb K$ always exists.

 Before proceeding,  we need to bring to the front the notion of the canonical pre-order of a monoid.
Let ${\mathbb K}=(K,+,0)$ be a positive commutative monoid.
Consider the
binary relation~$\sqsubseteq$ on~$K$ defined as follows: for all~$a,b \in K$, we have that $a \sqsubseteq b$ if and only if there exists some~$c\in K$ such that~$a
+ c = b$. It is easy to see that~$\sqsubseteq$ is a reflexive and
transitive relation; hence, it is  a pre-order, called the \emph{canonical pre-order} on $\mathbb K$ (see, \cite[Section 3.3]{gondran2008graphs}).

Let $R(X)$ and $S(X)$ be two
$\mathbb K$-relations over the  set $X$. We write $R\sqsubseteq S$ to denote that 
$R(t)\sqsubseteq S(t)$ holds,
for every $X$-tuple $t$. 
Clearly, this is a pre-order on $\mathbb K$-relations over $X$.

We are now ready to introduce the notion of a semijoin function on a monoid.

\begin{definition}\label{defn:semijoin} Let ${\mathbb K}=(K,+,0)$ be a positive  commutative monoid. A \emph{semijoin function on $\mathbb K$} is a function $\sj$  that takes as arguments two $\mathbb K$-relations $R(X)$ and $T(Y)$, and returns as value   a $\mathbb K$-relation $\sj(R,T)$ over $X$ with the following four properties:
\begin{description} 
    \item [\emph{(P1)}] If $R$ and $T$ are consistent, then
    $\sj(R,T) = R$.
    \item [\emph{(P2)}] $\sj(R,T)\sqsubseteq R$.
     \item [\emph{(P3)}] $\sj(R,T)[X\cap Y] \sqsubseteq T[X\cap Y]$.
     \item [\emph{(P4)}] If $T[X\cap Y]\sqsubseteq R[X\cap Y]$, then 
     $\sj(R,T)[X\cap Y]=T[X\cap Y]$.
\end{description}
We say that \emph{$\mathbb K$ has a semijoin function} if at least one semijoin function on $\mathbb K$ exists.
\end{definition}

\begin{example} \label{exam:cons-wit}
We now give an example and a  non-example of a semijoin function.
\begin{enumerate}
\item 
 If $\mathbb B =(\{0,1\}, \vee, 0)$ is the Boolean monoid, then 
 the semijoin operation $\ltimes$ on standard relations is 
 an example of a semijoin function on $\mathbb B$. Indeed, property $\sf{(P1)}$ holds for $\ltimes$ because if two standard relations $R(X)$ and $T(Y)$ are consistent, then their join $R\Join T$ witnesses their consistency, hence $R\ltimes T = (R\Join T)[X]= R$. For properties  $\sf{(P2)}$,  $\sf{(P3)}$,  
 $\sf{(P4)}$, first note that the relation $\sqsubseteq$ on  $\mathbb B$ is just
  set-theoretic containment  $\subseteq$. Properties $\sf{(P2)}$,  $\sf{(P3)}$,  
 $\sf{(P4)}$ can now be proved easily using the definitions of the semijoin  $\ltimes$ and the join $\Join$ of two standard relations.  For instance, towards proving property $\sf{(P4)}$, assume that $T[X\cap Y]\subseteq R[X\cap Y]$. We have to show that 
 $R \ltimes S[X\cap Y] = T[X\cap Y]$. In view of property $\sf{(P3)}$, it suffices to show that $T[X\cap Y]\subseteq (R\ltimes T)[X\cap Y]$. Take a tuple
 $t \in T[X\cap Y]$. Then there is a $Y$-tuple $t_1\in T $ such that $t_1[X\cap Y]=t$. Since $T[X\cap Y]\subseteq R[X\cap Y]$, we have that $t\in R[X\cap Y]$, hence there is a $X$-tuple $t_2\in R$ such that $t_2[X\cap Y]=t$. Therefore, $t_2$ joins with $t_1$, hence $t_2\in R\ltimes T$ and so $t=t_2[X\cap Y]\in (R\ltimes T)[X\cap Y]$. 
 \item If ${\mathbb N}=(N,+,0)$ is the bag monoid, then the semijoin $R\ltimes_{\mathbb N} T$ of two bags $R(X)$ and $T(Y)$  is not a semijoin function on $\mathbb N$.
 Here, 
 $R\ltimes_{\mathbb N} T = (R\Join_{\mathbb N} T)[X]$, where $R\Join_{\mathbb N} T$ is the \emph{bag-join} of $R$ and $T$, i.e., for every $(X\cup Y)$-tuple $t$, we have that
 $(R\Join_{\mathbb N} T)(t) = R(t[X])\cdot T(t[Y])$.
 The reason that $\ltimes_{\mathbb N}$ is not a semijoin function on $\mathbb N$ is that $\ltimes_{\mathbb N}$ does not have  property $\sf{(P1)}$. Indeed, as pointed out in \cite{DBLP:conf/pods/AtseriasK21}, the bag-join of two consistent bags need not witness their consistency, hence property $\sf{(P1)}$ fails for such bags.  In particular, this is the case for the bags $R(A,B) = \{ (1,2):1, (2,2):1\}$ and $T(B,C) = \{(2,1):1, (2,2):1\}$. They are consistent but it is easy to check that $R\ltimes_{\mathbb N} T=\{(1,2):2, (2,2):2\}\not = R$. 
\end{enumerate}
\end{example}

As defined in \cite{DBLP:journals/sigmod/AtseriasK25},
a \emph{consistency witness function on $\mathbb K$} is a   function $W$ that takes as arguments two  $\mathbb K$-relation $R(X)$  and  $T(Y)$, and returns as value a $\mathbb K$-relation $W(R,T)$ over  $X\cup Y$ such that if $R$ and $T$ are consistent $\mathbb K$-relations, then $W(R,T)$ is a consistency witness for $R$ and $T$.  For example, the join $R\Join T$ of two standard relations $R(X)$ and $T(Y)$ is a consistency witness function on $\mathbb B$. 

\begin{definition} \label{defn:coherent}
Let $\mathbb K$ be a positive commutative monoid and let $W$ be a consistency witness function on $\mathbb K$. We say that $W$ is \emph{coherent} if the  function ${\mathcal S}_W$ 
is a semijoin function on $\mathbb K$, where ${\mathcal S}_W$ takes as input two $\mathbb K$-relations $R(X)$ and $T(Y)$, and returns as value the projection of $W(R,T)$ on $X$, that is, 
${\mathcal S}_W(R,T) =W(R,T)[X].$

We say that $\mathbb K$ \emph{has a coherent consistency witness function} if at least one coherent consistency witness function on $\mathbb K$ exists.
\end{definition}
Clearly, the join $\Join$ of two standard relations is a coherent consistency witness function on $\mathbb B$, 
since $\mathcal S_{\Join}$ is the semijoin function
$\ltimes$ on standard relations.
In contrast, let $W$ be the consistency witness function on standard relations such that
$W(R,T)= R\Join T$, if $R$, $T$ are consistent, and $W(R,T)=\emptyset$, if $R$, $T$ are not consistent. Then $W$ is not coherent because ${\mathcal S}_W$ does not have property ${\sf (P4)}$; observe, though,  that ${\mathcal S}_W$ has properties ${\sf (P1)}$, ${\sf (P2)}$, ${\sf (P3)}$.

The next result  gives a sufficient condition for the existence of an explicit coherent consistency witness function; hence, it also gives a sufficient condition for the existence of an explicit semijoin function. 

\begin{proposition} \label{prop:coherent}
Let ${\mathbb K}=(K,+,0)$ be a positive commutative monoid.
If $\mathbb K$ has an expansion to a bounded distributive lattice ${\mathbb K'}= (K,+,\times,0,1)$, then the
standard join $\Join_{{\mathbb K'},S}$ is a coherent consistency witness function on $\mathbb K$, where for $\mathbb K$-relations $R(X)$ and $T(Y)$ and for every $(X\cup Y)$-tuple $t$, we have that
$(R\Join_{{\mathbb K'},S} T)(t)=
R(t[X])\times T[t[Y]).$
\end{proposition}

\begin{proof} Let $R(X)$ and $T(Y)$ be two $\mathbb K$-relations.
To keep the notation light,  assume that $X=\{A,B\}$ and $Y=\{B,C\}$;  other than a heavier  notation,  the general case is proved the same way.

Assume that the monoid ${\mathbb K}=(K,+,0)$ has an expansion to a bounded distributive lattice ${\mathbb K'}=(K,+,\times, 0,1)$. Let $W$ be the standard join $\Join_{{\mathbb K'}, S}$, where for $\mathbb K$-relations $R(X)$ and $T(Y)$ and for every $(X\cup Y)$-triple $(a,b,c)$, we have that
$$W(R,S)(a,b,c)= (R\Join_{{\mathbb K'}, S} T)(a,b,c) = R(a,b) \times T(b,c).$$
As shown in \cite{AK25},  $W$ is a consistency witness function on $\mathbb K$. To show that 
$W$ is coherent, we must show that the function $\sj_W$ with 
$\sj_{W}(R,T)=  W(R,T)[X]$ is a semijoin function on $\mathbb K$, i.e., we must show that $\sj_W$ has properties $\sf{(P1)}$, $\sf{(P2)}$, $\sf{(P3)}$, and $\sf{(P4)}$ in Definition \ref{defn:semijoin}. Note that,  since $\mathbb K'$ is  a bounded distributive lattice,  the following facts hold:
\begin{enumerate}
\item 
For all $u \in K$, we have that $u\times u = u$.
\item For all $u,v \in K$, we have that 
$u \times v \sqsubseteq u$ and
$u \times v \sqsubseteq v$.
\item For all $u,v, w \in K$, if $u\sqsubseteq v$, then $u\times w \sqsubseteq v\times w$.
\end{enumerate}

\noindent Property $\sf{(P1)}$: $\sj_W$ has property $\sf{(P1)}$ because $W$ is a consistency witness function.

\smallskip

\noindent Property $\sf{(P2)}$: We must show that $\sj_W(R,T) \sqsubseteq R$, i.e.,   $\sj_W(R,T)(a,b) \sqsubseteq R(a,b)$ holds, for every $X$-pair $(a,b)$.
For every such pair $(a,b)$, we have that
$$\sj_W(R,T)(a,b) = W(R,T)[X](a,b)= \sum_c W(R,T)(a,b,c) = \sum_c R(a,b)\times T(b,c)= $$
$$= R(a,b)\times \big (\sum_c T(b,c)\big )~\sqsubseteq~ R(a,b),$$
where the above equalities hold by the definitions of $W$ and $\sj_W$, the definition of the projection operation, and the distributivity of $\mathbb K'$, while the last inequality $\sqsubseteq$ holds by fact 2  about the bounded distributive lattice $\mathbb K'$.

\smallskip

\noindent Property $\sf{(P3)}$: We must show that $\sj_W(R,T) [B] \sqsubseteq T[B]$, i.e., $\sj_W(R,T) [B](b)  \sqsubseteq T[B](b)$ holds, for every $B$-element $b$. For every such element $b$, we have that
$$\sj_W(R,T)[B](b) = (W(R,T)[X][B])(b)= W(R,T)[B](b)= 
\sum_{a,c} W(R,T)(a,b,c) = $$
$$\sum_{a,c} R(a,b)\times T(b,c) = 
\sum_a R(a,b) \times \big (\sum_c T(b,c)\big ) =
\big (\sum_a R(a,b)\big ) \times T[B](b) = $$
$$R[B](b) \times T[B] (b)~  
\sqsubseteq~  T[B](b),$$
where the above equalities hold by the definitions of $W$ and $\sj_W$, the definition of the projection operation,   part 2 of Proposition \ref{lem:easyfacts1}, and the distributivity of $\mathbb K'$, while the inequality $\sqsubseteq$ holds by  fact 2 about  the bounded distributive lattice $\mathbb K'$.

\smallskip

\noindent Property $\sf{(P4)}$: Assume that
$T[B] \sqsubseteq R[B]$ holds. We must show that
$\sj_W(R,T)[B] = T[B]$ holds as well. In view of property  $\sf{(P3)}$, it suffices to show that $T[B] \sqsubseteq \sj_W(R,T)[B]$, i.e, it suffices to show that $T[B](b)  \sqsubseteq \sj_W(R,T)[B](b)$ holds, for every $B$-element $b$. As seen in the proof of proof of property $\sf{(P3)}$, for every $B$-element $b$, we have that
$$\sj_W(R,T)[B](b) = R[B](b) \times T[B](b).$$ 
Since we assumed that 
$T[B] \sqsubseteq R[B]$ holds, we have that $T[B](b) \sqsubseteq R[B](b)$. 
Hence, by using the facts 1 and 3 about the bounded distributive lattice $\mathbb K'$, we have that
$$T[B](b) = T[B](b) \times T[B](b)~\sqsubseteq~ R[B](b) \times T[B](b) = \sj_W[B](b).$$
This completes the proof of property $\sf{(P4)}$ and 
the proof of the propositions.
\end{proof}

 \begin{example} \label{exam:coherent}
We now give two  examples that illustrate the applicability of Proposition \ref{prop:coherent}. 
 \begin{enumerate}
 \item For every set $A$, the powerset monoid ${\mathbb P}(A)= (\mathcal{P}(A), \cup,\emptyset)$ expands to the bounded distributive lattice $(\mathcal{P}(A), \cup,\cap, \emptyset, A)$.  In this case, the standard join is the coherent consistency witness function $W$ such that 
 $W(R,S)(t) = R(t[X]) \cap S(t[Y]).$
 \item The fuzzy monoid ${\mathbb V}=([0,1], \max, 0)$ expands to the bounded distributive lattice\\ $([0,1], \max, \min, 0, 1))$. In this case, the standard join is the coherent consistency witness function $W$ such that 
 $W(R,S)(t) = \min\{R(t[X]),S(t[Y])\}.$
 \end{enumerate}
 \commentout{
 \item The monoid  ${\mathbb R}^{\geq 0}=([0,\infty), +, 0)$ expands to the semifield $([0,\infty), +, \times, /,0, 1)$.
 In this case, the Vorob'ev join is the coherent consistency witness function $W$ such that 
$$W(R,S)(t) =
\begin{cases}
(R(t[X])\times S(t[Y]))/R(t[X\cap Y])  & \text{if $R(t[X\cap Y]) = T[t(X\cap Y)] \not = 0$}\\
			0 & \text{otherwise}.\\
		 \end{cases}$$
\item   The monoid ${\mathbb T}= ((-\infty,\infty], \min, \infty)$ expands to the semifield $((-\infty,\infty], \min, +, -,\infty,0)$. 
In this case, the Vorob'ev join is the coherent consistency witness function $W$ such that 
$$W(R,S)(t) =\begin{cases}
R(t[X])+ S(t[Y]) -R(t[X\cap Y]) &
\text{if $R(t[X\cap Y]) = T[t(X\cap Y)] \not = \infty$}
\\
			\infty & \text{otherwise}.\\
		 \end{cases}$$
 (in the previous two expressions, $\times$, $/$, $+$, and $-$ are the standard multiplication, division, addition, and subtraction operations on the real numbers).
}
 \end{example}
 

So far, we have not exhibited a semijoin function for  the bag monoid ${\mathbb N}=(N,+,0)$.  The existence of such a function on $\mathbb N$ will follow from the  results in this section. For now, we point out that there are positive commutative monoids for which no semijoin function exists.

\begin{proposition} \label{prop:<3,5>}

The numerical semigroup ${\mathbb K}_{3,5}=(\langle 3,5\rangle, +,0)$ has no semijoin function, where $\langle 3,5\rangle=\{3m+5n: m, n \geq 0\}$,   
\end{proposition}
\begin{proof} Towards a contradiction, assume that $\sj$ is a semijoin function on  ${\mathbb K}_{3,5}$.  Let $U$, $V$ be two attributes.
Let $R(U,V)$ be the ${\mathbb K}_{3,5}$-relation with $R(u_1,v)=3$, $R(u_2,v)=3$, and $R(z,w)=0$, for all other pairs. Let $T(V)$
be the ${\mathbb K}_{3,5}$-relation with $T(v)=5$ and $T(z)=0$, if $z\not = v$. Then $T[V] \sqsubseteq R[V]$, since $T[V]=T$ and $R[V]$ is the ${\mathbb K}_{3,5}$-relation with
$R[V](v)=6$ and $R[V](z)=0$, if $z\not = v$. Consider the value $\sj(R,T)$ of the semijoin function $\sj$ on the pair $(R,T)$. By property $\sf{(P4)}$ and since  $T[V] \sqsubseteq R[V]$, we must have that
$\sj(R,T)[V]=T[V]=T$. By property $\sf{(P2)}$, we must have that $\sj(R,T)\sqsubseteq R$, which is impossible because there do not exist two elements in the set $\langle 3,5\rangle$ such that each is at most $3$ and their sum is $5$.
\end{proof}

Our first goal is to characterize the positive commutative monoids for which a semijoin function exists. Towards this goal, we introduce a combinatorial property of monoids.

\begin{definition}\label{dfn:frugal}
Let ${\mathbb K}=(K,+,0)$ be a positive commutative monoid. 
\begin{itemize}
\item We say that $\mathbb K$ has the \emph{\fpp~for a natural number $n\geq 1$}
if for  all elements  $b, c_1, c_2, \ldots, c_n$  in $K$ such that $b\sqsubseteq c_1+c_2+\cdots +c_n$,  there are elements $d_1, d_2, \ldots, d_n$ in $K$ such that $d_1\sqsubseteq c_1, d_2 \sqsubseteq c_2, \ldots, d_n\sqsubseteq c_n$, and $d_1+d_2+\cdots +c_n = b$.
\item We say that $\mathbb K$ has the \emph{\fpp}~if it has the \fpp~for every natural number $n\geq 1$.
\end{itemize}
\end{definition}

Note that the \fpp~for $n=1$ is trivial.
The term ``\fpp'' was chosen because this property has a motivation from operations research. Consider a company that produces a  product at $n$ different locations and assume  that the maximum production capacity at location $i$ is $c_i$, $1\leq i\leq n$.  Suppose that the total demand for the product is $b$, where $b\sqsubseteq c_1+\cdots+c_n$.  The \fpp~for $n$ asserts that there is a production policy for the company so that the total production at the $n$ locations meets precisely the demand with no waste whatsoever (i.e., there is no excess production).

Observe that the proof of Proposition \ref{prop:<3,5>} actually shows that the numerical semigroup ${\mathbb K}_{3,5}$ does not have the \fpp~for $n=2$; this is witnessed by $c_1=3$, $c_2=3$, and $b=5$. As we shall see next, it is no accident that  ${\mathbb K}_{3,5}$ has no semijoin function and at the same time lacks the \fpp~for $n=2$.

\begin{theorem} \label{thm:fpp}
Let $\mathbb K$ be a positive commutative monoid. Then the following statements are equivalent:
\begin{enumerate}
\item $\mathbb K$ has a coherent consistency witness function.
\item $\mathbb K$ has a semijoin function.
\item $\mathbb K$ has the \fpp~for $n=2$
\item $\mathbb K$ has the \fpp.
\end{enumerate}
\end{theorem}
\begin{proof}
We will prove this theorem in a round-robin style.

\noindent  $(1) \Longrightarrow (2)$
Follows immediately from  Definition \ref{defn:coherent} of a coherent consistency witness function. 

\smallskip

\noindent  $(2) \Longrightarrow (3)$ Assume that  $\sj$ is a semijoin function on $\mathbb K$. Given elements $b, c_1, c_2$  in $K$ such that $b\sqsubseteq c_1+c_2$, consider the $\mathbb K$-relations $R(U,V)$
and $T(V)$ such that
\begin{itemize}
    \item $R(u_1,v)=c_1$, $R(u_2,v) =c_2$, and $R(z,w)=0$, for all other pairs.
    \item $T(v)=b$ and $T(z)=0$, for all $z\not = v$.
\end{itemize}
Since $b\sqsubseteq c_1 +c_2$, we have that
$T[V] \sqsubseteq R[V]$. Hence, by property $\sf{(P4)}$ of the semijoin function  $\sj$, we must have that
$\sj(R,T)[V]=T[V]=T$. By property $\sf{(P2)}$ of the semijoin function $\sj$, we must have that $\sj(R,T)\sqsubseteq R$. Hence, $\sj(R,T)$ must be a $\mathbb K$-relation 
of the form  $\sj(R,T)(u_1,v)=d_1$,  $\sj(R,T)(u_2,v)=d_2$, and $\sj(R,T)(z,w) = 0$, for all other pairs, where $d_1$, $d_2$ are elements of $K$ such that $d_1\sqsubseteq c_1$, $d_2\sqsubseteq c_2$. Furthermore, since $\sj(R,T)[V]=T[V]=T$, we must have that $d_1+d_2 = b$. Thus, $\mathbb K$ has the \fpp~for $n=2$.

\smallskip

\noindent $(3) \Longrightarrow (4)$ 
Assume that $\mathbb K$ has the \fpp~for $n=2$. We have to show that $\mathbb K$ has the \fpp~for every $k\geq 1$.  As mentioned earlier, the \fpp~for $k=1$ holds trivially.  By induction on $k$, assume that  $\mathbb K$ has the \fpp~ for $k$. We will show that $\mathbb K$ has the \fpp~for $k+1$. Towards this, assume that   $b, c_1, c_2, \ldots, c_k,c_{k+1}$ are elements   in $K$ such that $b\sqsubseteq c_1+c_2+\cdots +c_k+c_{k+1}$. We must show that  there are elements $d_1, d_2, \ldots, d_k,d_{k+1}$ in $K$ such that $d_1\sqsubseteq c_1, d_2 \sqsubseteq c_2, \ldots, d_k\sqsubseteq c_k, d_{k+1}\sqsubseteq c_{k+1}$, and $d_1+d_2+\cdots +d_k+d_{k+1} = b$.
Let $c=c_1+c_2\cdots+c_k$. Since $b\sqsubseteq c_1+c_2+\cdots +c_k+c_{k+1}$, we have that $b\sqsubseteq c+c_{k+1}$. Since $\mathbb K$ has the \fpp~for $n=2$, it follows that there are elements $d$ and $d_{k+1}$ in $K$ such that
$d\sqsubseteq c$, $d_{k+1}\sqsubseteq c_{k+1}$, and $b = d +d_{k+1}$. Since $c=c_1+c_2\cdots+c_k$, we have that $d\sqsubseteq c_1+c_2\cdots+c_k$. By induction hypothesis, $\mathbb K$ has the
\fpp~for $k$. Hence, there are elements
$d_1,d_2,\ldots,d_k$ in $K$ such that 
$d_1\sqsubseteq c_1, d_2 \sqsubseteq c_2, \ldots, d_k\sqsubseteq c_k$, and $d_1+d_2+\cdots +d_k = d$; consequently, 
$d_1+d_2+\cdots +d_k+d_{k+1} = d+d_{k+1} = b$. Thus, $\mathbb K$ has the \fpp~for $k+1$.

\smallskip

\noindent $(4) \Longrightarrow (1)$ Assume that
$\mathbb K$ has the \fpp. We will construct a coherent consistency witness function
$W$ on   $\mathbb K$ by considering cases for arbitrary input $\mathbb K$-relations $R(X)$ and $T(Y)$ to $W$. 
\smallskip

   \noindent{\emph{Case 1.}} The $\mathbb K$-relations $R(X)$ and $T(Y)$ are consistent. We pick a $(X\cup Y)$-relation $U$ witnessing the consistency of $R(X)$ and $T(X)$, and we set $W(R,T)=U$.
At this point, we already have that $W$ is a consistency witness on $\mathbb K$.  Thus, in the remainder of the proof, we will focus on showing in each case that the function $\sj_W$ is a semijoin function on $\mathbb K$, where $\sj_W(R,T)=W(R,T)[X]$.
In the present case, we have that
$\sj_W(R,T)=W(R,T)[X]=U[X]=R$, since $U$ witnesses the consistency of $R$ and $T$. Therefore, properties $\sf{(P1)}$ and
$\sf{(P2)}$ in Definition \ref{defn:semijoin} hold. Properties $\sf{(P3)}$ and $\sf{(P4)}$ in Definition \ref{defn:semijoin}  also hold because, since $R(X)$ and $T(Y)$ are consistent, they are also inner consistent. This means that $R[X\cap Y]=T[X\cap Y]$, hence $\sj_W(R,T)[X\cap Y]=W(R,T)[X][X\cap Y]=R[X\cap Y]=T[X\cap Y].$

\smallskip

   \noindent{\emph{Case 2.}} The $\mathbb K$-relations $R(X)$ and $T(Y)$ are not  consistent. We need to define the value of $W(R,T)$ in such a way that $\sj_W(R,T)=W(R,T)[X]$ satisfies properties $\sf(P2)$, $\sf(P3)$, $\sf(P4)$. We distinguish two subcases.

\smallskip

   \emph{Subcase 2.1.} Assume that $T[X\cap Y]\not \sqsubseteq R[X\cap Y]$. In this case, we set $W(R,T)=\emptyset$, i.e., $W(R,T)(w)=0$, for every
   $(X\cup Y)$-tuple $w$. It follows that 
   $\sj_W(R,T)=\emptyset$, i.e., $\sj_W(R,T)(t) = 0$, for every $X$-tuple $t$. Clearly, properties $\sf(P2)$, $\sf(P3)$, $\sf(P4)$ are satisfied.

\smallskip

 \emph{Subcase 2.2.} Assume that $T[X\cap Y] \sqsubseteq R[X\cap Y]$. 
We will show that there is a $\mathbb K$-relation $V$
over $X\cup Y$ such that $V[X]\subseteq R$
 and $V[X\cap Y]=  T[X\cap Y]$.  This $\mathbb K$-relation $V$ will be the desired value
$W(R,T)$, as then $V[X]$ will satisfy properties 
$\sf{(P2)}$, $\sf{(P3)}$, and $\sf{(P4)}$.
 Let $t$ be a $(X\cap Y)$-tuple in the support of $T[X\cap Y]$ and let
$T[X\cap Y](t)=b\not = 0$. Let $t_1,t_2, \ldots,t_n$ be the $X$-tuples in the support of $R$ such that $t_1[X]=t,t_2[X]=t, \ldots,t_n[X]=t$. Let $c_1=R(t_1),c_2=R(t_2),\ldots, c_n=R(t_n)$. Since
$T[X\cap Y] \sqsubseteq R[X\cap Y]$, we have that $b\leq c_1 + c_2 +  \cdots +c_n$.  Since $\mathbb K$ has the \fpp, there are elements $d_1,d_2,\ldots,d_n$ in $K$ such that $d_1\sqsubseteq c_1, d_2\sqsubseteq c_2,\ldots, d_n\sqsubseteq c_n$ and $d_1+d_2+\cdots + c_n=b$.  Let $v$ be a tuple such that  $tv$ is in the support of $T$, where $tv$ is the tuple obtained by concatenating $t$ and $v$.
We now start building the $\mathbb K$-relation $V$ on $X\cup  Y$ by
putting $P(t_1v)=d_1,P(t_2v)=d_2, \ldots, P(t_nv)=d_n$, where $t_iv$ is the tuple obtained by concatenating
$t_i$ and $v$, $1\leq i\leq n$. We continue this process by going over each of the remaining tuples in the support of $T[X\cap Y]$, applying the transportation property in each case, and progressively building the relation $V$. We also put $V(w)=0$, for every $(X\cup Y)$-tuple $w$ such that
$w[X\cap Y]$ is not in the support of $T[X\cap Y]$.  By construction, the $\mathbb K$-relation $V$ obtained this way is such that $V[X]\sqsubseteq R$ (because of the inequalities of the form $d_{1}\sqsubseteq c_j$, for $1\leq j\leq n$) and $V[X\cap Y]= T[X\cap Y]$ (because of the equalities of the form $d_{1}+ d_{2}+\cdots+d_{n}=b$). Finally, we set
$W(R,T)=V$.
\end{proof}

Theorem \ref{thm:fpp} provides a useful tool for determining the existence of semijoin functions.  As an illustration, 
for every $k\geq 3$, let
${\mathcal P}_k$ be the set consisting of the empty set $\emptyset$, the $k$-element set $\{1,\ldots,k\}$, and all $(k-1)$-element subsets of $\{1,\ldots,k\}$. The \emph{truncated powerset} ${\mathbb P}_k$ is the structure ${\mathbb P}_k= ({\mathcal P}_k, \cup, \emptyset)$. For example,
$\mathbb{P}_3 = (\{\emptyset,\{1,2\},\{1,3\},\{2,3\},\{1,2,3\}\},\cup,\emptyset)$.
Clearly, each truncated powerset ${\mathbb P}_k$ is a positive commutative monoid. We claim that no truncated powerset ${\mathbb P}_k$ has a semijoin function.
To see this, consider the sets $B=\{1,3,\ldots,k-1\}$, $C_1=\{1,3,\ldots,k\}$, $C_2=\{2,3,\ldots,k\}$. Clearly, $B\sqsubseteq C_1 \cup C_2$, but there are no sets $D_1$ and $D_2$ in ${\mathcal P}_k$ such that
$D_1\sqsubseteq C_1$, $D_2\sqsubseteq C_2$, and $ D_1\cup D_2= B$. Thus, the truncated powerset  ${\mathbb P}_k$ is a monoid that does not have a semijoin function.

A broader consequence of Theorem \ref{thm:fpp} is the decidability of the existence of a semijoin function for finite positive commutative monoids.

\begin{corollary} \label{cor:decide}
The following problem is decidable in cubic time: given a finite positive commutative monoid $\mathbb K$, does $\mathbb K$ have a semijoin function?
In fact, this problem is decidable in \emph{LOGSPACE}.
\end{corollary}
\begin{proof}
By Theorem \ref{thm:fpp}, this problem is equivalent to checking that $\mathbb K$ has the \fpp~for $n=2$; this can be checked in  time 
$O(|K|^3)$ by examining all triples with elements in  $K$ one at a time. Clearly, this algorithm can be implemented in logarithmic space. 
\end{proof}

Our next result asserts that  the inner consistency property implies the existence of a semijoin function.

\begin{theorem} \label{thm:semijoin-exist}
    Let $\mathbb K$ be a positive commutative monoid.
If $\mathbb K$ has the inner consistency property, then $\mathbb K$ has a  semijoin function.
\end{theorem}
\begin{proof}
     As shown in  \cite{AK25}, the inner consistency property is  equivalent to the \ftp,  defined as follows. We say that a positive commutative monoid
${\mathbb K}= (K, +,0)$
has the  \emph{\ftp}~if
 for every $m\geq 1$, every $n\geq 1$,
 every~$m$-vector~$b = (b_1,\ldots,b_m) \in K^m$  
 and
 every~$n$-vector~$c = (c_1,\ldots,c_n) \in K^n$  
 such that~$b_1 + \cdots + b_m = c_1 + \cdots + c_n$ holds,  there
is  an~$m \times n$
matrix~$D = (d_{ij} : i \in [m], j \in [n]) \in K^{m \times n}$  whose 
rows sum to $b$ and whose columns sum to~$c$,
 i.e.,~$d_{i1} + \cdots + d_{in} = b_i$ for all~$i \in [m]$
 and~$d_{1j} + \cdots + d_{mj} = c_j$ for all~$j \in [n]$.
 Theorem 3 in \cite{AK25} asserts that for a positive commutative monoid $\mathbb K$, the following statements are equivalent:
 \begin{enumerate}
 \item $\mathbb K$ has the inner consistency property.
 \item $\mathbb K$ has the \ftp.
 \end{enumerate}
We will now show that the transportation property implies  the \fpp.

Let $\mathbb K$ be a positive commutative monoid that has the \ftp.
Assume that $n\geq 1$  and  $b, c_1,c_2,\ldots,c_n$ are elements in $K$ such that
$b\sqsubseteq c_1+c_2+\cdots,+c_n$.
Therefore, there is an element
   $b'\in K$ such that $b+b'=c_1+ c_2 + \cdots + c_n$.
Consider the vectors $(b,b')$ and  $(c_1,c_2,\ldots,c_n)$. Since $\mathbb K$ has the transportation property, there are elements $d_{ij}\in K$, where $i=1,2$ and $1\leq j\leq n$, such that the following equations hold:
\begin{equation*}
\begin{array}{ccccccccc}
d_{11} & + & d_{12} & + & \cdots & + & d_{1n} & =      & b  \\
+      &   & +      &   &        &   & +      &        &      \\
d_{21} & + & d_{22} & + & \cdots & + & d_{2n} & =      & b'  \\
\shortparallel  &   & \shortparallel     &   &        &   & \shortparallel     &        & \\
c_1    &   & c_2    &   &        &   & c_n   &        &
\end{array}
\end{equation*}
Since $d_{1j}+d_{2j}=c_j$ for $1\leq j \leq n$, we have that $d_{1,j}\sqsubseteq c_j$ for $1\leq j\leq n$. Thus, the elements $d_{1,1}, d_{1,2},\ldots, d_{1n}$ witness that $\mathbb K$ has the \fpp.

We can now put everything together: if $\mathbb K$ has the inner consistency, then, by Theorem 3 in \cite{AK25}, $\mathbb K$ has the \ftp. Therefore, by what  was shown above, $\mathbb K$ has the \fpp. Hence, by Theorem \ref{thm:fpp}, $\mathbb K$ has a semijoin function.
\end{proof}

As immediate consequences of Theorem \ref{thm:semijoin-exist}, 
the bag monoid ${\mathbb N}=(N,+,0)$, the monoid
${\mathbb R}^{\geq 0}=([0,\infty), +, 0)$, and the monoid  ${\mathbb T}= ((-\infty,\infty], \min, \infty)$
have  semijoin functions, because these three monoids are known to have  the inner consistency property \cite{DBLP:conf/pods/AtseriasK21}. 

By Theorem \ref{thm:semijoin-exist}, the inner consistency property implies the existence of a semijoin function. The converse is not true.
Indeed, 
let ${\mathbb N}_2=(\{0,1, 2\}, \oplus, 0)$ be the positive commutative monoid in which $0$ is the neutral element and
$1\oplus 1 = 2 = 1 \oplus 2 = 2 \oplus 1 = 2\oplus 2$. In \cite{AK25}, it was shown that ${\mathbb N}_2$ does not have the inner consistency property. It is easy to see that 
${\mathbb N}_2$ has the \fpp~for  $n=2$, so,  by Theorem \ref{thm:fpp},  ${\mathbb N}_2$
has a semijoin function. Thus, the inner consistency property is stricter than the existence of a semijoin function.

The last result in this section asserts that the existence of a semijoin function and the \emph{cancellativity} property imply the inner consistency property.

\begin{itemize}
    \item  We say that a  monoid
    ${\mathbb K}=(K,+, 0)$ is \emph{cancellative} if for all elements $a,b,c$ in $K$ with $a+b=a+c$ , we have that $b=c$.
\end{itemize}

\begin{proposition} \label{prop:cancel}
Let ${\mathbb K}=(K,+,0)$ be a positive commutative monoid. If $\mathbb K$ is cancellative and has a semijoin function, then $\mathbb K$ has the inner consistency property.
\end{proposition}
\begin{proof}
By Theorem \ref{thm:semijoin-exist} and Theorem 3 in \cite{AK25}, it suffices to show that if $\mathbb K$ is cancellative and has the \fpp, then $\mathbb K$ has the transportation property.  It is also known 
(see \cite[Proposition 10]{AK25} and \cite[page 508]{DBLP:books/daglib/0023547})
that $\mathbb K$ has the transportation property if and only if $\mathbb K$ has the $2\times 2$ transportation property, that is, for all  $(b_1,b_2) \in K^2$ and for $(c_1,c_2) \in K^2$ with $b_1+b_2=c_1+c_2$, there are elements $d_{11},d_{12}, d_{21}, d_{22}$ in $k$ such that
\begin{align*}
d_{11}  +  d_{12} & =  b_1 & 
d_{21}  +  d_{22} &=   b_2 &
d_{11}  +  d_{21} & =   c_1 &
d_{12}  +  d_{22} &=   c_2.
\end{align*}
So, assume that $\mathbb K$ is cancellative and has the \fpp. We will show that $\mathbb K$ has the $2\times 2$ transportation property. Towards this, assume that $b_1,b_2,c_1,c_2$ are elements in $K$ such that $b_1+b_2 = c_1 +c_2$. It follows that $b_1\sqsubseteq c_1+c_2$, hence, by the \fpp, there are elements
$d_1,d_2$ in $\mathbb K$ such that $d_1\sqsubseteq c_1$, $d_2\sqsubseteq c_2$, and $d_1+d_2= b_1$. Since $d_1\sqsubseteq c_1$ and $d_2\sqsubseteq c_2$, there are elements $e_1,e_2$ in $K$ such that $d_1+e_1=c_1$ and $d_2+e_2=c_2$. Since $b_1+b_2 = c_1 +c_2$, we have that
$d_1+d_2+b_2 = d_1+e_1+d_2+e_2= d_1+d_2+e_1+e_2$. Since $\mathbb K$ is cancellative, it follows that $b_2=e_1+e_2$. Therefore, the desired elements $d_{ij}$ are
$d_{11}=d_1$, $d_{12}=d_2$, 
$d_{21}=e_1$, $d_{22}=e_2$.
\end{proof}

 Proposition \ref{prop:cancel}  yields a useful tool for identifying monoids that do not have a semijoin function;   the following is vast generalization of Proposition  \ref{prop:<3,5>}.

\begin{corollary} \label{cor:numerical} 
${\mathbb N}=(N,+,0)$ is the only  numerical semigroup that has a semijoin function.
\end{corollary}
\begin{proof}
As shown in \cite{AK25}, the bag monoid $\mathbb N$ is the only numerical semigroup that has the inner consistency property. However, every numerical semigroup is cancellative. 
\end{proof}

\section{Semijoin Programs and Full Reducers} \label{sec:full-reducer}

 Semijoin programs and  full reducers in relational databases were first studied in \cite{DBLP:journals/jacm/BernsteinC81,DBLP:journals/siamcomp/BernsteinG81,BeeriFaginMaierYannakakis1983}.  
\begin{itemize}
\item 
A \emph{semijoin program on a schema $X_1,\ldots,X_m$} is a finite sequence of assignment statements of the form
    $X_i := X_i \ltimes X_j$, where $i,j$ are indices from $[1,m]$.
    \item Let $\pi$ be a semijoin program on $X_1,\ldots,X_m$  and let $R_1(X_1),\ldots,R_m(X_m)$ be a collection of standard relations. We write $R^*_1(X_1),\ldots,R^*_m(X_m)$ to denote the collection of standard relations resulting from $\pi$, when each $X_i$ is interpreted by the standard relation  $R_i(X_i)$.
    \item A semijoin program $\pi$ on a hypergraph $H$ with  $X_1,\ldots,X_m$ as its hyperedges is a \emph{full reducer on $H$} if for every collection $R_1(X_1),\ldots,R_m(X_m)$ of standard relations, the resulting collection
$R^*_1(X_1),\ldots,R^*_m(X_m)$  is globally consistent.
\end{itemize}
Beeri et al.~\cite{BeeriFaginMaierYannakakis1983} characterized acyclicity in terms of the existence of full reducers. Earlier, a variant of this result was obtained by Bernstein and Goodman \cite{DBLP:journals/siamcomp/BernsteinG81}.

\begin{theorem}[\cite{BeeriFaginMaierYannakakis1983}] \label{thm:semijoin-standard}
For every hypergraph $H$, the following statements are equivalent:
\begin{enumerate}
\item $H$ is acyclic.
\item $H$ has a full reducer.
\end{enumerate}
\end{theorem}

We will investigate the question: does Theorem \ref{thm:semijoin-standard} extend to annotated relations? 
In what follows, we assume that ${\mathbb K}=(K,+,0)$ is a positive commutative monoid. 

\begin{definition}\label{defn:sj-program}
Let $X_1,\ldots,X_m$ be a schema
and let $\sjk$ be a binary function symbol.
\begin{itemize}
    \item A \emph{semijoin program on $X_1,\ldots,X_m$} is a finite sequence of assignment statements of the form
    $X_i := X_i \sjk X_j$, where $i,j$ are indices from $[1,m]$.
    \item Let $\pi$ is a semijoin program on $X_1,\ldots,X_m$, let $\sj$ be a semijoin function on $\mathbb K$, and let \\$R_1(X_1),\ldots,R_m(X_m)$ be a collection of $\mathbb K$-relations. We write $R^*_1(X_1),\ldots,R^*_m(X_m)$ to denote the collection of $\mathbb K$-relations resulting from $\pi$ when $\sjk$ is interpreted by the semijoin function $\sj$ and each $X_i$ is interpreted by the $\mathbb K$-relation $R_i(X_i)$. 
\end{itemize}
\end{definition}

Clearly, the evaluation of a semijoin program on a collection of $\mathbb K$-relations is meaningful only when a semijoin function on $\mathbb K$ exists.
\begin{example}
Consider the semijoin program

\centerline{$X_3 :=  X_3 \sjk X_2;~~ 
X_2 :=  X_2 \sjk X_3;~~ 
X_1 : =   X_1 \sjk X_2.$}

If $R_1(X_1), R_2(X_2), R_3(X_3)$ are three $\mathbb K$-relations and $\sj$ is a semijoin function on $\mathbb K$, then the preceding semijoin program produces the relations  $R_1^*(X_1), R_2^*(X_2), R_3^*(X_3)$, where

\centerline{$R_3^* :=  \sj(R_3,R_2);~~ 
R_2^* :=  \sj(R_2,\sj(R_3,R_2)); ~~
R_1^* : =  \sj(R_1, \sj(R_2,\sj(R_3,R_2))).$}
\end{example}

 \begin{definition} \label{defn:full-reducer}
 Let $H$ be a hypergraph with 
 $X_1,\ldots, X_m$ as its hyperedges. 
 \begin{itemize}
\item Let $\pi$ be a semijoin program on $X_1,\ldots,X_m$ and let $\sj$ be a semijoin function on $\mathbb K$.

The pair $(\pi,\sj)$ is a \emph{full reducer for $H$} if for every collection $R_1(X_1),\ldots,R_m(X_m)$ of $\mathbb K$-relations, the resulting collection
$R^*_1(X_1),\ldots,R^*_m(X_m)$ of $\mathbb K$-relations is globally consistent.
\item $H$ has a \emph{full reducer on $\mathbb K$} if there are a
semijoin program $\pi$ on $H$ and a 
semijoin function $\sj$ function on $\mathbb K$ such that the pair $(\pi,\sj)$ is a full reducer for $H$.
\end{itemize}
\end{definition}

 We are now ready to state and prove the main result of this section.

 \begin{theorem} 
 \label{thm:full-reducer}
Let $\mathbb K$ be a positive commutative monoid that has the inner consistency property.  For every hypergraph $H$, 
 the following statements are equivalent:
\begin{enumerate} 
\item $H$ is acyclic.
\item  There is a semijoin program $\pi$ on $H$ that depends only on $H$ such that for every semijoin function $\sj$ on $\mathbb K$, the pair $(\pi,\sj)$ is a full reducer for $H$.
\item $H$ has a full reducer on $\mathbb K$.
\end{enumerate}
\end{theorem}
\begin{proof}
We will prove this theorem in a round-robin style.

\noindent $(1) \Longrightarrow (2)$
Since $H$ is an acyclic hypergraph, $H$ has the running intersection property (see Theorem \ref{thm:BFMY}). This means that there is an ordering 
$X_1,\ldots,X_m$ of the hyperedges of $H$ so that the following holds:
for every $i$ with $1< i \leq m$, there is an index 
$j_i < i$ such that
$(X_1 \cup \cdots \cup X_{i-1})\cap X_i \subseteq X_{j_i}$.  Note that if $i=2$, then $j_i=1$, but if 
$i>2$, then $j_i$ can be any index smaller than $i$; for example, if $i=3$, then $j_3$ can be $1$ or $2$. Let $\pi$ be the following semijoin program on $\mathbb K$, where each assignment statement is labeled by a number:
\begin{align*}
(-m) & ~~~ X_{j_m}  := X_{j_m} \sjk X_m \\
(-m+1) & ~~~  X_{j_{m-1}}  :=   X_{j_{m-1} } \sjk X_{m-1}\\
\cdots & \\
(-k) & ~~~ X_{j_k} := X_{j_k} \sjk X_k \\
\cdots & \\
(-3) & ~~~ X_{j_3} :=  X_{j_3} \sjk X_3 \\
(-2) & ~~~ X_1 :=  X_1 \sjk X_2 \\
(2) & ~~~ X_2 := X_2 \sjk X_1 \\
(3) & ~~~ X_3 :=  X_3 \sjk X_{j_3} \\
\cdots & \\
(k) & ~~~ X_k  := X_k \sjk X_{j_k} \\
\cdots & \\
(m-1) & ~~~  X_{m-1} : = X_{m-1} \sjk X_{j_{m-1}}  \\
(m) & ~~~ X_m := X_m \sjk  X_{j_m}. 
\end{align*}
Let $\sj$ be a semijoin function on $\mathbb K$, let $R_1(X_1),\ldots,R_m(X_m)$ be a collection of $\mathbb K$-relations, and let $R^*_1(X_1),\ldots,R^*_m(X_m)$ be the collection of $\mathbb K$-relations obtained by applying $\pi$ to the semijoin function $\sj$ and the collection $R_1(X_1),\ldots,R_m(X_m)$. 
We have to show that the collection $R^*_1(X_1),\ldots,R^*_m(X_m)$ is globally consistent. We start with a crucial stepping stone.

\smallskip

\noindent{\emph{Claim 1.}} For every $k$ with
$2\leq k \leq m$, the $\mathbb K$-relations 
$R^*_k(X_k)$ and $R^*_{j_k}(X_{j_k})$ are consistent.

By the inner consistency of $\mathbb K$, it suffices to show that
$R^*_k[X_{j_k}\cap X_k] = R^*_{j_k}[X_{j_k}\cap X_k]$. If $p$ is one of the labels of the assignment statements of $\pi$, we will write $R^{(p)}_i$ to denote the $\mathbb K$-relation over $X_i$ after the $p$-th line of $\pi$ has been applied. In particular, we have that $R_i^{(m)}(X_i)= R_i^*(X_i)$, for all $i$ with $1\leq i\leq m$. Consider the assignment statement $X_{j_k}  := X_{j_k} \sjk X_k$ of $\pi$ labeled by $(-k)$.  Thus,  
\begin{equation}R^{(-k)}_{j_k}=\sj(R^{(-k+1)}_{j_k},R^{(-k+1)}_k).
\end{equation}
By property
$\sf{(P3)}$ of the semijoin function $\sj$,
we have that
\begin{equation} \sj(R^{(-k+1)}_{j_k},R^{(-k+1)}_k)[X_{j_k}\cap X_k]\sqsubseteq R^{(-k+1)}_k[X_{j_k}\cap X_k].
\end{equation}
From (1) and (2), it follows that  
\begin{equation}R^{(-k)}_{j_k}[X_{j_k}\cap X_k]   \sqsubseteq   R^{(-k+1)}_k[X_{j_k}\cap X_k]. 
\end{equation}
Since $j_i < i$ holds for every $i$ with $2\leq i\leq m$,  the set $X_i$ does not appear in the left-hand side of any assignment statement of $\pi$ with labels strictly between $-k$ and $k$. Thus,
\begin{equation}
R_k^{(-k+1)} = R_k^{(k-1)}.
    \end{equation}
From (3) and (4), it follows that
\begin{equation}R^{(-k)}_{j_k}[X_{j_k}\cap X_k]   \sqsubseteq   R^{(k-1)}_k[X_{j_k}\cap X_k]. 
\end{equation}
Property $\sf{(P2)}$ of the semijoin function $\sj$ asserts that $\sj(R,T)\sqsubseteq R$, for all $\mathbb K$-relations $R$ and $T$.  From this property and  the fact that every assignment statement of $\pi$ is of the form
 $X_i  := X_i \sjk X_{j_i}$, it follows that if $s \leq t$, then $R^{(t)}_i \sqsubseteq R^{(s)}_i$ holds for every $i$ with $1\leq i\leq m$.   In particular, 
 $R_{j_k}^{(k-1)}\sqsubseteq R_{j_k}^{(-k)}$ holds. Consequently, we have that
 \begin{equation}
 R_{j_k}^{(k-1)}[X_{j_k}\cap X_k]\sqsubseteq R_{j_k}^{(-k)}[X_{j_k}\cap X_k].
 \end{equation}
 From (5), (6), and the transitivity of the pre-order $\sqsubseteq$, it follows that 
 \begin{equation}
 R_{j_k}^{(k-1)}[X_{j_k}\cap X_k]\sqsubseteq R_k^{(k-1)}[X_{j_k}\cap X_k].
 \end{equation}
 Consider now the assignment statement $X_k:= X_k \sjk X_{j_k}$ of $\pi$ labeled by $(k)$. Thus,
 \begin{equation}
R_k^{(k)}=\sj(R_k^{(k-1)},R_{j_k}^{(k-1)}).
 \end{equation}
By property $\sf {(P4)}$ of 
$\sj$, if  $R$ and $T$ are two
$\mathbb K$-relations such that  $T[X\cap Y]\sqsubseteq R[X\cap Y]$, then 
     $\sj(R,T)[X\cap Y]=T[X\cap Y]$.
     From (7), (8), and this property, it follows that
     \begin{equation}
     R^{(k)}_k[X_{j_k}\cap X_k] = R_{j_k}^{(k-1)}[X_{j_k}\cap X_k].
     \end{equation}
     Observe that none of the assignment statements of $\pi$ that come after the assignment statement with label $(k)$ has $X_k$ of $X_{j_k}$ in its left-hand side. Consequently,  $R^*_k= R^{(m)}_k= R^{(k)}_k$ and $R^*_{j_k} = R^{(m)}_{j_k} = R^{(k-1)}_{j_k}$. From these two facts and (9), it follows that
     \begin{equation}
     R^*_k[X_{j_k}\cap X_k] = R_{j_k}^*[X_{j_k}\cap X_k],
     \end{equation}
     which establishes Claim 1.
We are now ready to complete the proof of the theorem.

\smallskip

\noindent{\emph{Claim 2.}} For every $k$ with $2\leq k\leq m$, the collection
$R_1^*(X_1),\ldots,R_k^*(X_k)$ of $\mathbb K$-relations is globally consistent.

This claim will be proved by induction on $k$.

\noindent{\emph{Base Case.}}  Claim 2 is true for $k=2$, because the $\mathbb K$-relations $R_2^*(X_1)$ and
$R^*_{j_2}(X_{j_2})$ are consistent by Claim 1 and, as pointed out earlier, $j_2=1$.

\noindent{\emph{Inductive Step.}} Assume that the collection 
$R_1^*(X_1),\ldots,R_1^*(X_k)$ is globally consistent, where $2< k < m$.  We have to show that the collection
$R_1^*(X_1),\ldots,R_1^*(X_k),R_{k+1}^*(X_{k+1})$
is globally consistent as well.
Let $W(X_1\cup \cdots \cup X_k)$ be a $\mathbb K$-relation that witnesses the global consistency of the collection
$R_1^*(X_1),\ldots,R_1^*(X_k)$. This means that $W[X_j]=R_j(X_j)$, for $1\leq j\leq k$.
We will first show that the $\mathbb K$-relations $W(X_1\cup \cdots \cup X_k)$ and
$R_{k+1}^*(X_{k+1})$ are consistent.  By the inner consistency of $\mathbb K$, it suffices to show that
\begin{equation}
W[(X_1\cup \cdots \cup X_k)\cap X_{k+1}]= 
R_{k+1}[(X_1\cup \cdots \cup X_k)\cap X_{k+1}].
\end{equation}
In what follows, we will use  the second part of Proposition \ref{lem:easyfacts1}, which asserts that
for every $\mathbb K$-relation $R(X)$
and for all $Z \subseteq Y \subseteq X$, we have $R[Y][Z] = R[Z]$. Since the ordering $X_1,\ldots,X_m$ has the running intersection property, there is an index $l \leq k$ with the property that $(X_1\cup \cdots X_k)\cap X_{k+1} \subseteq X_l$. 
We now have
\begin{equation}
W[(X_1\cup \cdots X_k)\cap X_{k+1}]= W[X_l][(X_1\cup \cdots X_k)\cap X_{k+1}]= R^*_l[(X_1\cup \cdots X_k)\cap X_{k+1}],
\end{equation}
where the first equality follows from part 2 of Proposition \ref{lem:easyfacts1} and the second equality follows from the fact that $W[X_l]=R^*_l(X_l)$, since $l\leq k$ and $W$ witnesses the global consistency of the collection $R^*_1(X_1),\ldots,R^*_k(X_k)$. By Claim 1, we have that the $\mathbb K$-relations
$R^*_{k+1}(X_{k+1})$ and $R_l^*(l)(X_l)$ are consistent, hence they are also inner consistent, which means that
\begin{equation}
R_{k+1}^*[X_{k+1}\cap X_l] = R^*_l[X_{k+1}\cap X_l].
\end{equation}
Since  $(X_1\cup \cdots X_k)\cap X_{k+1} \subseteq X_l$, we have that also
  $(X_1\cup \cdots X_k)\cap X_{k+1} \subseteq X_{k+1}\cap X_{l}$ holds.
  From this fact and equation (13), we obtain
  \begin{equation}
R^*_{k+1}[(X_1\cup \cdots X_k)\cap X_{k+1}] = 
R^*_l[(X_1\cup \cdots X_k)\cap X_{k+1}].
  \end{equation}
From equations (12) and (14), we obtain
\begin{equation}
W[(X_1\cup \cdots X_k)\cap X_{k+1}] = 
R^*_{k+1}[(X_1\cup \cdots X_k)\cap X_{k+1}],
  \end{equation}
  which establishes that the $\mathbb K$-relations $W(X_1\cup \cdots\cup X_k)$
  and $R_{k+1}(X_{k+1})$ are consistent.
Let $U(X_1\cup \ldots \cup X_{k+1})$ be a $\mathbb K$-relation that witnesses the consistency of the $\mathbb K$-relations
$W(X_1\cup \cdots\cup X_k)$
  and $R_{k+1}(X_{k+1})$.  We claim that
  $U(X_1\cup \ldots \cup X_{k+1})$ witnesses the global consistency of the collection $R_1^*(X_1),\ldots,R_1^*(X_k),R_{k+1}^*(X_{k+1})$. Indeed, if
  $l\leq k$, we have
  \begin{equation}
  U[X_l] = U[X_1\cup \cdots \cup X_k][X_l]=
  W[X_1\cup \cdots X_k][X_l]=R^*_l(X_l), 
  \end{equation}
  where the first equality follows from part 2 of Proposition \ref{lem:easyfacts1}, the second equality holds because $U(X_1\cup \cdots \cup X_{k+1})$ witnesses the consistency of $W(X_1\cup \cdots \cup X_k)$ and $R_{k+1}(X_{k+1})$, and the third equality  holds because $l\leq k$ and $W(X_1\cup \cdots \cup X_k)$ witnesses the global consistency of the collection $R_1^*(X_1),\ldots,R_1^*(X_k)$.  Finally, $U[X_{k+1}]=R^*_{k+1}(X_{k+1})$, because  $U(X_1\cup \cdots \cup X_{k+1})$ witnesses the consistency of $W(X_1\cup \cdots \cup X_k)$ and $R_{k+1}(X_{k+1})$.
  This completes the proof of Claim 2.  By applying Claim 2 to $k=m$, we see that the collection $R_1^*(X_1),\ldots,R_m^*(X_m)$ is globally consistent, hence the pair $(\pi, \sj)$ is a full reducer for $H$.

\smallskip

\noindent $(2) \Longrightarrow (3)$
This direction is obvious (note also that at least one semijoin function on $\mathbb K$ exists because $\mathbb K$ has the inner consistency property).

\smallskip

\noindent $(3) \Longrightarrow (1)$  Assume that $H$ has a full reducer on $\mathbb K$. We will show that $H$ has the \ltgc~for
$\mathbb K$-relation.
Let $X_1,\ldots,X_m$ be the hyperedges of  $H$.  Let $R_1(X_1),\ldots, R_m(X_m)$ be a collection of pairwise consistent $\mathbb K$-relations. We claim that this collection is globally consistent. Let $\pi$ be a semijoin program on $H$ and let $\mathcal S$ be a semijoin function on $\mathbb K$ such that the pair $(\pi,\sj)$ is a full reducer for $H$. It follows that  the resulting collection $R^*_1(X_1),\ldots, R^*_m(X_m)$ of $\mathbb K$-relations is globally consistent. Now, take any assignment statement $R_i:= \cj(R_i,R_j)$ in the semijoin program $\pi$. Since $R_i$ and $R_j$ are consistent $\mathbb K$-relations, property $\sf{(P1)}$ of the semijoin function $\sj$ implies that $\sj(R_i,R_j)=R_i$. Thus, $\pi$ and $\sj$ do not change the input relations
$R_1(X_1)\ldots, R_m(X_m)$, which means that  $R_i^*(X_i) = R_i(X_i)$ holds, for every $i$ with $1\leq i\leq m$. Therefore,   $R_1(X_1),\ldots, R_m(X_m)$  is a globally consistent collection of $\mathbb K$-relations.   Since $H$ has the \ltgc~for $\mathbb K$-relations, 
direction $(2)\Longrightarrow (1)$ of Theorem \ref{thm:BFMY-general} implies that $H$ is acyclic.
\end{proof}

Let's discuss the similarities and the differences between the proof of Theorem \ref{thm:semijoin-standard} in \cite{BeeriFaginMaierYannakakis1983}  and the proof of the preceding Theorem \ref{thm:full-reducer}. As regards similarities,  the shape of the semijoin program $\pi$ is the same in both proofs. Furthermore, the proof of Claim 1 in the proof of Theorem \ref{thm:fpp} follows the same steps  as the proof of a corresponding claim in the proof of Theorem \ref{thm:semijoin-standard}. As regards differences between the two proofs, the standard semijoin operation is used in the proof of Theorem \ref{thm:semijoin-standard}, while here an arbitrary semijoin function on $\mathbb K$ is used. 
Furthermore, once Claim 1 has been established, the proof of Theorem \ref{thm:full-reducer} is quite different from that of Theorem \ref{thm:semijoin-standard}. Indeed, in the  proof of Theorem \ref{thm:semijoin-standard}, the join operation $\Join$ on standard relations is used, while here we have to use consistency witnesses, since, as seen earlier, even the bag-join need not be a witness to the consistency of two bags. Note also that the inner consistency property is used throughout the proof of Theorem \ref{thm:full-reducer}.

The proof of Theorem \ref{thm:full-reducer}  pinpoints the structural properties of an abstract semijoin operation needed to establish the equivalence between acyclicity and the existence of a full reducers for annotated relations. It also brings out the importance of the canonical preorder $\sqsubseteq$ on a monoid, as the properties 
$\sf{(P2)}$, $\sf{(P3)}$, and $\sf{(P4)}$ of a semijoin function involve the canonical preorder in a crucial way.
Finally, Theorem \ref{thm:full-reducer} establishes the rather remarkable fact that for every acyclic hypergraph $H$, a \emph{single}
     semijoin program is  a full reducer for $H$  uniformly  for all  commutative monoids $\mathbb K$ that have the inner consistency property and for all semijoin functions on $\mathbb K$.


Our final result characterizes the inner consistency property  in terms of full reducers.

\begin{corollary} \label{thm:semijoin-char}
Let $\mathbb K$ be a positive commutative monoid. Then the following statements are equivalent:
\begin{enumerate}
\item $\mathbb K$ has the inner consistency property.
\item Every acyclic hypergraph has a full reducer on $\mathbb K$,
\item The $3$-path hypergraph $P_3$ has a full reducer on $\mathbb K$.
\end{enumerate}
\end{corollary}

\begin{proof}
We will prove this result in a round-robin style.

\noindent{$(1)\Longrightarrow (2)$} This follows from direction
$(1)\Longrightarrow (2)$ of Theorem \ref{thm:full-reducer}.

\smallskip

\noindent{$(2)\Longrightarrow (3)$} This is true because that  $3$-Path hypergraph $P_3$ is acyclic.

\smallskip

\noindent{$(3)\Longrightarrow (1)$}  Suppose that  the $3$-Path hypergraph $P_3$ has a full reducer on $\mathbb K$.   The proof of direction $(3)\Longrightarrow (4)$ of Theorem \ref{thm:full-reducer} shows that $P_3$ has the \ltgc~for $\mathbb K$-relations; note that the inner consistency property is not used in that proof, only property $\sf{(P1)}$ of a semijoin function and the definition of a full reducer are used.  Then, direction $(3)\Longrightarrow (1)$ of Theorem 4 implies that  $\mathbb K$ has the inner consistency property.
    \end{proof}

\section{Concluding Remarks}

The work reported here contains both conceptual and technical contributions. The main conceptual contribution is the introduction of the notion of a semijoin function on a monoid, which, in turn,  makes it possible to introduce the notion of a full reducer for a schema on a monoid. The main technical contributions are the characterization of the monoids that have a semijoin function and the characterization of acyclicity in terms of the existence of full reducers on monoids.  The latter  characterization vastly generalizes
the classical result in Beeri et al.\ \cite{BeeriFaginMaierYannakakis1983} and also yields new results for a variety of annotated relations over monoids.

In the Introduction, we  referenced some early influential work on semijoins. This is only a small fraction of the investigation of the semijoin operation on relational databases. There is, for example, a body of work on the semijoin algebra, which is the fragment of relational algebra in which the semijoin is used instead of the join \cite{DBLP:phd/be/LEINDERS08,DBLP:journals/jolli/LeindersMTB05,DBLP:journals/ipl/LeindersTB04}. In particular, it has been shown that the semijoin algebra has the same expressive power as the well-studied guarded fragment of first-order logic. More recently, semijoins have been used as a message-passing mechanism in query processing \cite{DBLP:conf/pods/KhamisNR16} and have also found applications in incremental view maintenance \cite{DBLP:journals/pvldb/WangHDY23}. Furthermore, advanced filtering techniques based on Bloom filters have been used in query processing as a generalization of semijoins \cite{DBLP:conf/cidr/YangZYK24}.
We leave it as future work to explore some of these topics in the context of annotated relations over monoids.

\smallskip

\noindent{\bf Acknowledgments} The work reported here builds on earlier work on annotated relation with Albert Atserias, including the study of the interplay between local and global consistency for annotated relations and the study of consistency witness functions for annotated relations. Many thanks to Albert for insightful comments on an early draft of this paper.


\commentout{Thoughts for future work:

\begin{itemize}
\item Show that if a positive commutative monoid is totally canonically pre-ordered and has an expansion to a semifield, then the Vorob'ev join is a coherent consistency witness function, where in the definition of the Vorob'ev join, the expression $\max_{\sqsubseteq}(R[X\cap Y], T[X\cap Y])$ has to be used in the denominator.

\item Reconcile the notion of a full reducer in Bernstein and Goodman with the notion of a full reducer in Beeri et al. The notes in Notability from the trip to Chicago in January 2026 contain a proof that a full reducer in the sense of Bernstein-Goodman is a full reducer in the sense of Beeri et al.  The other direction is not obvious. One should also verify that the class of tree queries in Bernstein and Goodman coincides with acyclic joins in the sense of Beeri et al.
\end{itemize}
}

\bibliography{biblio}

\end{document}